\documentclass{amsart}
\usepackage{graphicx}
\usepackage{graphicx}
\usepackage{amssymb}
\usepackage{amsfonts}
\usepackage{amsmath}
\usepackage{amsthm}
\newtheorem{thm}{Theorem}[section]
\newtheorem{dfn}[thm]{Definition}
\newtheorem{lem}[thm]{Lemma}
\newtheorem{rem}[thm]{Remark}
\newtheorem{cor}[thm]{Corollary}
\newtheorem{prop}[thm]{Proposition}

\numberwithin{equation}{section}
\newcommand{\cA}{\mathcal{A}}

\newcommand{\cF}{\mathcal{F}}

\newcommand{\cI}{\mathcal{I}}

\newcommand{\cS}{\mathcal{S}}

\begin{document}
\title{Perfect and Partial Hedging for Swing Game
Options in Discrete Time.}
 \author{Yan Dolinsky, Yonathan Iron and Yuri Kifer\\
 Institute of Mathematics\\
 Hebrew University of Jerusalem\\
 Jerusalem, Israel }%

\address{
 Institute of Mathematics, The Hebrew University, Jerusalem 91904, Israel\\
 {e.mail: yann1@math.huji.ac.il, yoni@math.huji.ac.il, kifer@math.huji.ac.il}}
\thanks{ Partially supported by the ISF grant no. 130/06}
 \date{\today}
\begin{abstract}\noindent
The paper introduces and studies hedging for game (Israeli) style extension
of swing options considered as multiple exercise derivatives. Assuming that
the underlying security can be traded without restrictions we derive a
formula for valuation of multiple exercise options via classical hedging
arguments. Introducing the notion of the shortfall risk for such options we
study also partial hedging which leads to minimization of this risk.
\end{abstract}

\subjclass[2000]{Primary: 91B28 Secondary: 60G40, 91B30}%
\keywords{hedging, multiple exercise derivatives, game options, shortfall risk.}%
\maketitle
\markboth{Y.Dolinsly, Y.Iron and Y.Kifer}{Hedging of swing options}
\renewcommand{\theequation}{\arabic{section}.\arabic{equation}}
\pagenumbering{arabic}

\section{Introduction}\label{sec:1}\setcounter{equation}{0}
Swing contracts emerging in energy and commodity markets (see
\cite{BG} and \cite{JRT}) are often modeled by multiple exercising
of American style options which leads to multiple stopping problems
(see, for instance, \cite{MH}, \cite{CT} and \cite{RZ}). Most
closely such models describe options consisting of a package of
claims or rights which can be exercised in a prescribed (or in any)
order with some restrictions such as a delay time between successive
exercises. Observe that peculiarities of multiple exercise options
are due only to restrictions such as an order of exercises and a
delay time between them since without restrictions the above claims
or rights could be considered as separate options which should be
dealt with independently. 

Attempts to valuate swing options in
multiple exercises models are usually reduced to maximizing the
total expected gain of the buyer which is the expected payoff in the
corresponding multiple stopping problem deviating from what now
became classical and generally accepted methodology of pricing
derivatives via hedging and replicating arguments. This digression
is sometimes explained by difficulties in using an underlying
commodity in a hedging portfolio in view of the high cost of
storage, for instance, in the case of electricity. We will not
discuss here in depth practical possibilities of hedging in energy 
markets but only observe that the seller of a swing option could, 
for instance, use for hedging certain securities linked to a corresponding 
commodity (electricity, gas, oil etc.) index. Another instrument
which can be used for hedging is an appropriate basket of stocks 
of major companies in the corresponding branch whose profit depends
in a computable way from the price of commodity in question. Though
such indirect hedging may seem to be not very precise it may still be helpful
taking into account that all duable mathematical models of financial markets
cannot describe them precisely and are used usually only as an auxiliary tool.
Another theoretical but may be not very realistic in practice possibility
is to buy from (and sell to) power stations an extra capacity
for electricity production instead of storing electricity itself and
use it as the underlying risky security for a hedging portfolio. We observe
also that multiple exercise options may appear in their own rights when an
investor wants to buy or sell an underlying security in several instalments
at times of his choosing. Anyway, the study of
hedging for multiple exercise options is sufficiently motivated from
the financial point of view and it leads to interesting mathematical
problems. In this paper we assume that the underlying security can
be used for construction of a hedging portfolio without restrictions
as in the usual theory of derivatives and, moreover, we will deal
here with the more general game (Israeli) option (contingent claim)
setup when both the buyer (holder) and the seller (writer) of the
option can exercise or cancell, respectively, the claims (or rights)
in a given order but as in \cite{Ki} each cancellation entails a
penalty payment by the seller. This required us, in particular, to
extend Dynkin's games machinery to the multiple stopping setup.

 In this paper a discrete time swing (multi stopping) game option is a
contract between its seller and the buyer which allows to the seller to
cancel (or terminate) and to the buyer to exercise $L$ specific
claims or rights in a particular order. Such contract is determined
given $2L$ payoff processes $X_i(n)\geq Y_i(n)\geq 0,\, n=0,1,...$,
$i=1,2,...,L$ adapted to a filtration $\cF_n,\, n\geq 0$ generated
by the stock (underlying risky security) $S_n,\, n\geq 0$ evolution.
If the buyer exercises the $k$-th claim $k\leq L$ at the time $n$ then
the seller pays to him the amount $Y_k(n)$ but if the latter cancels the
claim $k$ at the time $n$ before the buyer
he has to pay to the buyer the amount $X_k(n)$ and the difference
$\delta_k(n)=X_k(n)-Y_k(n)$ is viewed as the cancelation penalty. In
addition, we require a delay of one unit of time between successive
exercises and cancellations. Observe that unlike some other papers (cf. 
\cite{CT}) we allow payoffs depending on the exercise number so, for instance,
our options may change from call to put and vice versa after different 
exercises.

The first goal of this paper is to develop a mathematical theory for
pricing of swing game options. The standard definition of the fair
price of a derivative security in a complete market is the minimal
initial capital needed to create a (perfect) hedging portfolio, and
so we have to start with a precise definition of a perfect hedge.
Observe that a natural definition of a perfect hedge in a multi
exercise framework is not a straightforward extension of a standard
one and it has certain peculiarities. Namely, the seller of the
option does not know in advance when the buyer will exercise the
$(j-1)$-th claim but his hedging strategy of the $j$-th claim should
depend on this (random) time and on the capital he is left with in
the portfolio after the $(j-1)$-payoff. Thus, in addition to
the usual dependence on the stock evolution a perfect hedge of the
$j$-th claim should depend on the past behavior of both seller and
the buyer of the option. Actually, an optimal portfolio allocation
depends also on the payoff processes of the future claims. The
construction of hedging strategies in the multiple exercise setup
requires a nontrivial additional iterative procedure in contrast to
the 1-exercise case where perfect hedging strategies are obtained 
directly from the martingale representation. 
Several papers dealt with mathematical analysis of swing American
options (see, for instance, \cite{CT} and \cite{RZ}) but none of
these papers defined explicitly what is a perfect hedge and what is
the option price. In \cite{RZ} the authors studied a specific type
of swing American options but they treated the problem from the
buyer point of view which in general is not interested in hedging
but only on a stopping strategy which will provide him a maximal
profit. In \cite{CT} the authors studied an optimal multi stopping
problem for continuous time models but they did not explained why
the value of the above problem under the martingale measure in a
complete market is the option price. In this paper we define the
notion of a perfect hedge for swing game options which generalize
swing American options, prove that in the binomial
Cox-Ross-Rubinstein (CRR) market the option price $V^*$ is equal to
the value of the multi stopping Dynkin game with discounted payoffs
under the unique martingale measure and provide a dynamical
programming algorithm which allows to compute both this value and a
corresponding perfect hedge. Similar
results can be obtained for the continuous time Black--Scholes
market with the stock price evolving according to the geometric
Brownian motion but in this paper we restrict ourselves to the
discrete time setup.

Our second  goal is to study hedging with risk for swing game
options. In real market conditions a writer of an option may not be
willing for various reasons to tie in a hedging portfolio the full
initial capital required for a perfect hedge. In this case the
seller is ready to accept a risk that his portfolio value will be
less than his obligation to pay and he will need additional funds to
fulfil the contract, i.e. the writer must add money to his portfolio
from other sources. In our setup the writer is allowed to add money
to his portfolio only at moments when the contract is exercised. The
shortfall risk is defined as the expectation with respect to the
market probability measure of the total sum that the seller added
from other sources. We will show that for any initial capital
$x<V^{*}$ there exists a hedge which minimizes the shortfall risk
and this hedge can be computed by a dynamical programming algorithm.
Observe that the existence of a hedge minimizing the shortfall risk
is not known in the continuous time even for usual (one stopping)
game options (see \cite{DK}). Hedging with risk was not studied
before for swing options of any type.

In Section 2 we define explicitly the notions of perfect and
partial hedges (the latter, for the shortfall risk case). Relying on
these we define the option price and the shortfall risk.
Then we state Theorem \ref{thm2.1} which yields the option price
together with the corresponding perfect hedge. Next, we formulate
Theorem \ref{thm2.2} which for a given initial capital provides the
shortfall risk and the corresponding optimal hedge together with the
dynamical programming algorithm for their computation.  In Section 3
we derive auxiliary lemmas needed in the proof, introduce the concept
of multi stopping Dynkin game and  prove existence of a
saddle point for this game. Section 4 and Section 5 are devoted to
the proofs of Theorem \ref{thm2.1} and Theorem \ref{thm2.2},
respectively.

\section{Preliminaries and main results}\label{sec:2}\setcounter{equation}{0}
Let $\Omega={\{1,-1\}}^N$ be the space of finite sequences
$\omega=(\omega_1,\omega_2,...,\omega_N)$; $\omega_i\in{\{1,-1\}}$
with the product probability $P={\{p,1-p\}}^N$, $p>0$. Consider the
binomial model of a financial market which is active at times
$n=0,1,...,N<\infty$ and it consists of a savings account $B_n$ with
an interest rate $r$ which without loss of generality (by discounting)
we assume to be zero, i.e.
\begin{equation}\label{2.1}
B_n=B_0>0,
\end{equation}
and of a stock whose price at time n equals
\begin{equation}\label{2.2}
 S_n=S_0\prod_{i=1}^n{(1+\rho_i)},\,\,\, S_0>0
\end{equation}
where
$\rho_i(\omega_1,\omega_2,...,\omega_N)=\frac{a+b}{2}+\frac{b-a}{2}{\omega_i}$
and $-1<a<0<b$. Thus $\rho_i$, $i=1,...,N$ form a sequence of independent
identically distributed (i.i.d.) random variables on the probability space
$(\Omega,P)$ taking values $b$ and $a$ with probabilities $p$ and $1-p$,
respectively. Recall, that the binomial CRR model is complete (see \cite{SKKM})
and $S_n,\, n\geq 0$ is a martingale with respect to the filtration
$\mathcal{F}_n=\sigma\{\rho_k,\, k\leq n\},\,\mathcal{F}_0=
\{\emptyset,\Omega\}$ and the unique martingale measure is given by
$\tilde{P}=\{\tilde{p},1-\tilde{p}\}^N$ where $\tilde{p}=\frac{a}{a-b}$.

We consider a swing option of the game type which has the $i$-th payoff,
$i\geq 1$ having the form
\begin{equation}\label{2.3}
H^{(i)}(m,n)=X_i(m)\mathbb{I}_{m<n}+Y_i(n)\mathbb{I}_{n\leq{m}},
\,\,\,\forall{m,n}
\end{equation}
where $X_i(n),Y_i(n)$ are $\mathcal{F}_n$-adapted and
$0\leq Y_i(n)\leq X_i(n)<\infty$. Thus for any $i,n$ there
exist functions
$f^{(i)}_n,g^{(i)}_n:\{a,b\}^n\rightarrow{\mathbb{R}_{+}}$ such that
\begin{equation}\label{2.4}
Y_i(n)=f^{(i)}_n(\rho_1,...,\rho_n), \
X_i(n)=g^{(i)}_n(\rho_1,...,\rho_n).
\end{equation}
For any $1\leq{i}\leq{L-1}$ let $C_i$ be the set of all pairs
$((a_1,...,a_i),(d_1,...,d_i))\in {\{0,...,N\}}^i\times {\{0,1\}}^i$
such that $a_{j+1}\geq {N\wedge{(a_j+1)}}$ for any $j<i$. Such
sequences represent the history of payoffs up to the $i$-th one in
the following way. If $a_j=k$ and $d_j=1$ then the seller canceled
the $j$-th claim at the moment $k$ and if $d_j=0$ then the buyer
exercised the $j$-th claim at the moment $k$ (maybe together with
the seller). For $n\geq{1}$ denote by $\Gamma_n$ the set of all
stopping times with respect to the filtration
${\{\mathcal{F}_n\}}_{n=0}^N$ with values from $n$ to $N$ and set
$\Gamma=\Gamma_0$.

\begin{dfn}\label{dfn2.1}
A stopping strategy is a sequence $s=(s_1,...,s_L)$  such that
$s_1\in\Gamma$ is a stopping time and for $i>1$,
$s_i:C_{i-1}\rightarrow{\Gamma}$ is a map which satisfies
$s_i((a_1,...,a_{i-1}),(d_1,...,d_{i-1}))\in
\Gamma_{N\wedge{(1+a_{i-1})}}$.
\end{dfn}

In other words for the $i$-th payoff both the seller and the buyer
choose stopping times taking into account the history of payoffs so
far. Denote by $\mathcal{S}$ the set of all stopping strategies and
define the map
$F:\mathcal{S}\times\mathcal{S}\rightarrow{\Gamma^{L}\times\Gamma^{L}}$
by $F(s,b)=((\sigma_1,...,\sigma_L),(\tau_1,...,\tau_L))$ where
$\sigma_1=s_1$, $\tau_1=b_1$ and for $i>1$,
\begin{eqnarray}\label{2.4+}
&\sigma_i=s_i((\sigma_1\wedge\tau_1,...,\sigma_{i-1}\wedge\tau_{i-1}),
(\mathbb{I}_{\sigma_1<\tau_1},...,
\mathbb{I}_{\sigma_{i-1}<\tau_{i-1}}))\,\,\,\mbox{and}\\
&\tau_i=b_i((\sigma_1\wedge\tau_1,...,\sigma_{i-1}\wedge\tau_{i-1}),
(\mathbb{I}_{\sigma_1<\tau_1},...,
\mathbb{I}_{\sigma_{i-1}<\tau_{i-1}})). \nonumber
\end{eqnarray}
Set
\begin{equation}\label{2.5}
c_k(s,b)=\sum_{i=1}^L \mathbb{I}_{\sigma_i\wedge\tau_i\leq{k}}
\end{equation}
which is a random variable equal to the number of payoffs until the
moment $k$.

For swing options the notion of a self financing portfolio involves not
only allocation of capital between stocks and the bank account but also
payoffs at exercise times. At the time $k$ the writer's decision how
much money to invest in stocks (while depositing the remaining money into
a bank account) depends not only on his present portfolio value but also
on the current claim. Denote by $\Xi$ the set of functions on the (finite)
probability space $\Omega$.
\begin{dfn}\label{dfn2.2}
A portfolio strategy with an initial capital $x>0$ is a pair
$\pi=(x,\gamma)$ where
$\gamma:{\{0,...,N-1\}}\times{\{1,...,L\}}\times\mathbb{R}\rightarrow{\Xi}$
is a map such that $\gamma(k,i,y)$ is an $\mathcal{F}_k$-measurable
random variable which represents the number of stocks which the seller
buy at the moment $k$ provided that the current claim has the number $i$ and
the present portfolio value is $y$. At the same time the sum
$y-\gamma(k,i,y)S_k$ is deposited to the bank account of the portfolio. We
call a portfolio strategy $\pi=(x,\gamma)$ \textit{admissible} if for any
$y\geq{0}$,
\begin{equation}\label{2.5+}
-\frac{y}{S_kb}\leq\gamma(k,i,y)\leq -\frac{y}{S_ka}.
\end{equation}
For any $y\geq{0}$ denote $K(y)=[-\frac{y}{b},-\frac{y}{a}]$.
\end{dfn}

Notice that if the portfolio value at the moment $k$ is $y\geq{0}$
then the portfolio value at the moment $k+1$ before the payoffs (if
there are any payoffs at this time) is given by
$y+\gamma(k,i,y)S_k(\frac{S_{k+1}}{S_k}-1)$ where $i$ is the number
of the next payoff. In view of independency of
$\frac{S_{k+1}}{S_k}-1$ and $\gamma(k,i,y)S_k$ we conclude that the
inequality (\ref{2.5+}) is equivalent to the inequality
$y+\gamma(k,i,y)S_k(\frac{S_{k+1}}{S_k}-1)\geq{0}$, i.e. the
portfolio value at the moment $k+1$ before the payoffs is nonnegative.
Denote by $\mathcal{A}(x)$ be the set of all \textit{admissible} portfolio
strategies with an initial capital $x>0$. Denote
$\mathcal{A}=\bigcup_{x>0}\mathcal{A}(x)$. Let $\pi=(x,\gamma)$ be a
portfolio strategy and $s,b\in\mathcal{S}$. Set
$((\sigma_1,...,\sigma_L),(\tau_1,...,\tau_L))=F(s,b)$ and
$c_k=c_k(s,b)$. The portfolio value at the moment $k$ after the
payoffs (if there are any payoffs at this moment) is given
by
\begin{eqnarray}\label{2.6}
&V^{(\pi,s,b)}_0=x-H^{(1)}(\sigma_1,\tau_1)
\mathbb{I}_{\sigma_1\wedge\tau_1=0} \ \ \mbox{and} \ \mbox{for} \
k>0, \\
&V^{(\pi,s,b)}_k=V^{(\pi,s,b)}_{k-1}+\mathbb{I}_{c_{k-1}<L}[\gamma(k-1,
c_{k-1}+1,V^{(\pi,s,b)}_{k-1})(S_k-S_{k-1})-
\nonumber\\
&\sum_{i=1}^L
H^{(i)}(\sigma_i,\tau_i)\mathbb{I}_{\sigma_i\wedge\tau_i=k}].\nonumber
\end{eqnarray}

\begin{dfn}\label{dfn2.3}
A perfect hedge is a pair $(\pi,s)$ which consists of a portfolio
strategy and a stopping strategy such that $V^{(\pi,s,b)}_k\geq{0}$
for any $b\in\mathcal{S}$ and $k\leq{N}$.
\end{dfn}
Observe that if $(\pi,s)$ is a perfect hedge then without loss of generality
we can assume that $\pi$ is an \textit{admissible} portfolio strategy and
throughout this paper we will consider only \textit{admissible} portfolio
strategies. As usual, the option price $V^*$ is defined as the infimum of
$V\geq{0}$ such that there exists a perfect hedge with an initial capital $V$.

The following theorem provides a dynamical programming algorithm for computation
of both the option price and the corresponding perfect hedge.

\begin{thm}\label{thm2.1}
Denote by $\tilde{E}$ the expectation with respect to the unique
martingale measure $\tilde{P}$. For any $n\leq{N}$ set
\begin{equation}\label{2.7}
X^{(1)}_n=X_L(n),\ Y^{(1)}_n=Y_L(n)\ and\  V^{(1)}_n=\min_{\sigma\in{\Gamma_{n}}}\max_{\tau\in{\Gamma_{n}}}
\tilde{E}(H^{(L)}(\sigma,\tau)|\mathcal{F}_n)
\end{equation}
and for $1<k\leq{L}$,
\begin{eqnarray}\label{2.8}
&X^{(k)}_n=X_{L-k+1}(n)+\tilde{E}(V^{(k-1)}_{(n+1)\wedge{N}}|\mathcal{F}_n),
\\
&Y^{(k)}_n=Y_{L-k+1}(n)+\tilde{E}(V^{(k-1)}_{(n+1)\wedge{N}}|\mathcal{F}_n)
\ \mbox{and} \nonumber\\
&V^{(k)}_n=\min_{\sigma\in{\Gamma_{n}}}\max_{\tau\in{\Gamma_{n}}}
\tilde{E}(X^{(k)}_{\sigma}\mathbb{I}_{\sigma<\tau}+Y^{(k)}_{\tau}
\mathbb{I}_{\sigma\geq\tau}
|\mathcal{F}_n).\nonumber
\end{eqnarray}
Then
\begin{equation}\label{2.8+}
V^{*}=V^{(L)}_0=\min_{s\in\mathcal{S}}\max_{b\in\mathcal{S}}G(s,b)
\end{equation}
where $G(s,b)=\tilde{E}\sum_{i=1}^L H^{(i)}(\sigma_i,\tau_i)$ and
$((\sigma_1,...,\sigma_L),(\tau_1,...,\tau_L))=F(s,b)$.
Furthermore, the stopping strategies
$s^{*}=(s^{*}_1,...,s^{*}_L)\in{S}$ and $b=(b^{*}_1,...,b^{*}_L)$
given by
\begin{eqnarray}\label{2.9}
&s^{*}_1=N\wedge{\min{\{k|X^{(L)}_k=V^{(L)}_k\}}}, \
b^{*}_1={\min{\{k|Y^{(L)}_k=V^{(L)}_k\}}},
\\
&s^{*}_i((a_1,...,a_{i-1}),(d_1,...,d_{i-1}))=N\wedge{\min{\{k>a_{i-1}|}}
\nonumber\\
&{{X^{(L-i+1)}_k=V^{(L-i+1)}_k\}}}, \
b^{*}_i((a_1,...,a_{i-1}),(d_1,...,d_{i-1}))\nonumber\\
&=N\wedge {\min{\{k>a_{i-1}| Y^{(L-i+1)}_k=V^{(L-i+1)}_k\}}}, \ i>1
\nonumber
\end{eqnarray}
satisfy
\begin{equation}\label{2.9+}
 G(s^{*},b)\leq G(s^{*},b^{*})\leq G(s,b^{*})\,\,\,
\mbox{for all}\,\, s,b
\end{equation}
and there exists a portfolio strategy
$\pi^{*}\in\mathcal{A}(V^{(L)}_0)$ such that $(\pi^{*},s^{*})$ is a
perfect hedge.
\end{thm}
Next, consider an option seller whose initial capital is $x$,
which is less than the option price, i.e. $x<V^*$. In this case the
seller must (in order to fulfill his obligation to the buyer) add
money to his portfolio from other sources. In our setup the seller
is allowed to add money to his portfolio only at times
when the contract is exercised. We also require that after the
addition of money by the seller the portfolio value must be positive.

\begin{dfn}\label{dfn2.4} An infusion of capital is a map
$I:{\{0,...,N\}}\times{\{1,...,L\}}\times\mathbb{R}\rightarrow{\Xi}$
such that $I(k,j,y)\geq{(-y)^{+}} $ is $\mathcal{F}_k$-measurable,
$I(k,L,y)=(-y)^{+}$ for any $k$, and for any $j<L$,
$I(N,j,y)=\big ((\sum_{i=j+1}^L Y_i(N))-y\big )^{+}$. The set of such maps
will be denoted by $\mathcal{I}$.
\end{dfn}
Thus  $I(k,j,y)$ is the amount that the seller adds to his portfolio
after the $j$-th payoff payed at the moment
$k$ and the portfolio value after this payment is $y$. When $k=N$ or
$j=L$ then clearly $I(k,j,y)$ is the minimal
amount which the seller should add in order to fulfill his
obligation to the buyer. Observe that when $k=N$ one
 infusion of capital to the seller's portfolio is already sufficient
in order to fulfill his obligations even if there are additional
payoffs at this moment, so we conclude that at each step that the
contract is exercised there is no more than one infusion of capital.
A hedge with an initial capital $x<V^{*}$ is a triple
$(\pi,\cI,s)\in\mathcal{A}(x)\times{I}\times{S}$ which consists of an
\textit{admissible} portfolio strategy with an initial capital $x$,
infusion of capital and a stopping strategy. Let $(\pi,I,s)$ be a
hedge and $b\in\mathcal{S}$ be a stopping strategy for the buyer.
Set $((\sigma_1,...,\sigma_L),(\tau_1,...,\tau_L))=F(s,b)$ and
$c_k=c_k(s,b)$. Define the stochastic processes
${\{W^{(\pi,I,s,b)}_k\}}_{k=0}^N$ and ${\{V^{(\pi,I,s,b)}_k\}}_{k=0}^N$
by
\begin{eqnarray}\label{2.11}
&W^{(\pi,I,s,b)}_0=x, \
V^{(\pi,I,s,b)}_0=x-\mathbb{I}_{\sigma_1\wedge\tau_1=0}\big (H^{(1)}
(\sigma_1,\tau_1)-\\
&I(0,1,x-H^{(1)}(\sigma_1,\tau_1))\big ) \ \mbox{and}\ \mbox{for} \ k>0,
\nonumber\\
&W^{(\pi,I,s,b)}_k=V^{(\pi,I,s,b)}_{k-1}+\mathbb{I}_{c_{k-1}<L}
\gamma(k-1,c_{k-1}+1,V^{(\pi,I,s,b)}_{k-1})(S_k-S_{k-1}),\nonumber\\
&V^{(\pi,I,s,b)}_k=W^{(\pi,I,s,b)}_k-\mathbb{I}_{c_{k-1}<L}
\mathbb{I}_{\sigma_{c_{k-1}+1}\wedge\tau_{c_{k-1}+1}=k}\times\nonumber\\
&\big (H^{(c_{k-1}+1)}(\sigma_{c_{k-1}+1},\tau_{c_{k-1}+1})+\mathbb{I}_{k=N}
\sum_{i=c_{k-1}+2}^L
Y_i(N) \nonumber\\
&-I(k,c_{k-1}+1,W^{(\pi,I,s,b)}_k-
H^{(c_{k-1}+1)}(\sigma_{c_{k-1}+1},\tau_{c_{k-1}+1}))\big
).\nonumber
\end{eqnarray}
Observe that if the contract was not exercised at a
moment $k$ then $W^{(\pi,I,s,b)}_k=V^{(\pi,I,s,b)}_k$ is the
portfolio value at this moment. If the contract was exercised at a
moment $k$ then $W^{(\pi,I,s,b)}_k$ and $V^{(\pi,I,s,b)}_k$ are the
portfolio values before and after the payoff, respectively.
Thus the total infusion of capital that made by the seller is given by
\begin{equation}\label{2.12}
C(\pi,I,s,b)=\sum_{i=1}^{(c_{N-1}+1)\wedge{L}}
I(\sigma_i\wedge\tau_i,i,W^{(\pi,I,s,b)}_{\sigma_i\wedge\tau_i}-H^{(i)}
(\sigma_1,\tau_i)).
\end{equation}
\begin{dfn}\label{dfn2.5}
Given a hedge $(\pi,I,s)\in\mathcal{A}\times{\cI}\times{S}$ the
shortfall risk for it is defined by
\begin{equation}\label{2.13}
R(\pi,I,s)=\max_{b\in\mathcal{S}}EC(\pi,I,s,b)
\end{equation}
which is the maximal expectation with respect to the market
probability measure $P$ of the total infusion of capital. The
shortfall risk for the intitial capital $x$ is defined by
\begin{equation}\label{2.14}
R(x)=\inf_{(\pi,I,s)\in \mathcal{A}(x)\times{\cI}\times{S}}R(\pi,I,s).
\end{equation}
\end{dfn}

The following result asserts for any initial capital $x$ there exists
a hedge  $(\pi,I,s)\in\mathcal{A}(x)\times{\cI}\times{S}$
which minimizes the shortfall risk and both the risk and the optimal
hedge can be obtained recurrently.

\begin{thm}\label{thm2.2}
Define a sequence of functions
$J_k:\mathbb{R}_{+}\times\{0,...,L\}\times{\{a,b}\}^k
\rightarrow{\mathbb{R}_{+}}$,
$0\leq{k}\leq{N}$ by the following formulas
\begin{eqnarray}\label{2.15}
&J_N(y,j,u_1,...,u_N)=((\sum_{i=L-j+1}^L
f^{(i)}_N(u_1,...,u_N))-y)^{+}, \ j>0,\\
&J_k(y,0,u_1,...,u_k)=0, \ 0\leq{k}\leq{N}\nonumber
\end{eqnarray}
and for $k<N$ and $j>0$,
\begin{eqnarray}\label{2.16}
&J_k(y,j,u_1,...,u_k)=\\
&\min\Bigg(\inf_{z\geq(g^{(L-j+1)}_k(u_1,...,u_k)-y)^{+}}\inf_{\alpha
\in{K(y+z-g^{(L-j+1)}_k(u_1,...,u_k))}}\nonumber\\
&\big (z+pJ_{k+1}(y+z-g^{(L-j+1)}_k(u_1,...,u_k)+b\alpha,j-1,u_1,...,u_k,b)+
\nonumber \\
&(1-p)J_{k+1}(y+z-g^{(L-j+1)}_k(u_1,...,u_k)+a\alpha,j-1,u_1,...,u_k,a)\big ),
\nonumber\\
&\max\Bigg(
\inf_{z\geq(f^{(L-j+1)}_k(u_1,...,u_k)-y)^{+}}\inf_{\alpha\in{K(y+z-
f^{(L-j+1)}_k(u_1,...,u_k))}}\nonumber\\
&\big (z+pJ_{k+1}(y+z-f^{(L-j+1)}_k(u_1,...,u_k)+b\alpha,j-1,u_1,...,u_k,b)+
\nonumber \\
&(1-p)J_{k+1}(y+z-f^{(L-j+1)}_k(u_1,...,u_k)+a\alpha,j-,u_1,...,u_k,a)\big ),
\nonumber\\
&\inf_{\alpha\in{K(y)}}\big (pJ_{k+1}(y+b\alpha,j,u_1,...,u_k,b)+\nonumber \\
&(1-p)J_{k+1}(y+a\alpha,j,u_1,...,u_k,a)\big )\Bigg)\Bigg).\nonumber
\end{eqnarray}
Then the shortfall risk for an initial capital $x$ is given by
\begin{equation}\label{2.17}
R(x)=J_0(x,L).
\end{equation}
Furthermore, the hedge
$(\tilde\pi=(x,\tilde\gamma),\tilde{I},\tilde{s})\in\mathcal{A}(x)\times{\cI}
\times{S}$
given by the formulas (\ref{5.31}), (\ref{5.34}) and (\ref{5.43}) satisfies
\begin{equation}
R(\tilde\pi,\tilde{I},\tilde{s})=R(x).
\end{equation}
\end{thm}

Not surprisingly the formulas above and their proof are quite technical
and complex since already for one stopping game options the corresponding
recurrent formulas for the shortfall risk in \cite{DK} and their proof
are rather complicated. Our method extends the approach of \cite{DK} by
relying on the dynamical programming algorithm for Dynkin's games with
appropriately modified payoff processes.

\begin{rem}\label{rem2.1}
Some applications may require a more general setup where the
first payoff is as before but the $i$-th payoff for $i>1$ depends also
on the first time when the $i$-th claim can be exercised, i.e. the $i$-th
payoff depends on the time of the $(i-1)$-th payoff.
The first payoff is exactly as in formula (\ref{2.3}). For $i>1$ we set
\begin{equation}\label{2.18}
\begin{split}
\forall{m,n\geq k}\ \ H^{(i,k)}(m,n)=X_{i,k}(m)
\mathbb{I}_{m<n}+Y_{i,k}(n)\mathbb{I}_{n\leq{m}}
\end{split}
\end{equation}
which is the $i$-th payoff if the seller cancells at time $m$ and the buyer
exercises at time $n$ provided the $i$-th claim can be exercised only starting
from the time $k$.
Here $X_{i,k}(n),Y_{i,k}(n)$ are $\mathcal{F}_n$-adapted stochastic processes
and $0\leq Y_{i,k}(n)\leq X_{i,k}(n)<\infty$. Definition \ref{dfn2.2} of a
portfolio strategy $\pi=(x,\gamma)$ with an initial capital $x$ should be also
modified so that $\gamma=\gamma(k,m,i,y)$ is an $\mathcal{F}_k$-measurable
random variable which represents the number of stocks which the seller
buy at the moment $m$ provided that the current claim which started at the
time $k\leq m$ has the number $i$ and the present portfolio value is $y$.
The definitions of perfect and partial hedges are the same as above. Then
we can obtain corresponding generalizations of Theorems \ref{thm2.1} and
\ref{thm2.2} whose proofs proceed similarly to the proof in Sections 4--5
but require an induction in an additional parameter which represents the time
of the previous payoff. Since the notations in this case are quite unwieldy
and the argument is longer but does not contain additional ideas we will not
deal with this generalization here.
\end{rem}

\section{Auxiliary lemmas}\label{sec:3}\setcounter{equation}{0}
The following lemma is a well known result about Dynkin games (see
\cite{YO}) which will be used for proving Theorems \ref{thm2.1} and
\ref{thm2.2}.
\begin{lem}\label{lem3.1}
Let ${\{X_n,Y_n\geq 0\}}_{n=0}^N$  be two adapted stochastic
processes. Set
\begin{equation*}
R(m,n)=\mathbb{I}_{m<n}X_{m}+\mathbb{I}_{m\geq{n}}Y_{n}
\end{equation*}
and define the stochastic process ${\{V_n\}}_{n=0}^N$ by
\begin{eqnarray*}
&V_N=Y_N, \ \mbox{and} \ \mbox{for} \ n<N\\
&V_n=Y_n\mathbb{I}_{Y_n>
X_n}+\min(X_n,max(Y_n,E(V_{n+1}|\mathcal{F}_n)))\mathbb{I}_{Y_n\leq
X_n}.
\end{eqnarray*}
Then
\begin{equation*}
V_n=\emph{ess-inf}_{\sigma\in\Gamma_{n}}\emph{ess-sup}_{\tau\in\Gamma_{n}}
E(R(\sigma,\tau)|\cF_n).
\end{equation*}
Moreover, for any stopping time $\theta\in\Gamma$ the stopping times
\begin{equation*}
\begin{split}
\sigma_{\theta}=\min\{k\geq\theta| X_{k}\leq V_k\}\wedge N \
\mbox{and} \ \tau_{\theta}=\min\{k\geq\theta| Y_{k}=V_k\}
\end{split}
\end{equation*}
satisfy
\begin{equation*}
E(R(\sigma_{\theta},\tau)|\mathcal{F}_{\theta})\leq V_{\theta}\leq
E(R(\sigma,\tau_{\theta})|\mathcal{F}_{\theta})
\end{equation*}
for any stopping times $\sigma,\tau\geq\theta$. Furthermore, for the
filtration $\{\cF_{(\theta+k)\wedge N }\}_{k=0}^N$ the processes
${\{V_{\sigma_{\theta}\wedge(\theta+k)\wedge N}\}}_{k=0}^N$,
${\{V_{\tau_{\theta}\wedge(\theta+k)\wedge N}\}}_{k=0}^N$ and
$V_{\sigma_{\theta}\wedge\tau_{\theta}\wedge(\theta+k)\wedge N}$,
are supermartingale, submartingale and martingale, respectively.
\end{lem}

Next, we derive auxiliary results which will be used for proving
Theorem \ref{thm2.1}. First, we generalize Dynkin games to the multi
stopping setup and show that also in this case there is a saddle
point, i.e., in particular, the multi stopping Dynkin game has a
value. Note that the following results about multi stopping
Dynkin's games are valid for any probability space with a discrete
finite filtration for which we use the same notations as before. The
main result concerning multi stopping Dynkin's games is the following.
\begin{prop}\label{prop3.1}
For any $s,b\in \cS$,
\begin{equation}\label{3.0}
G(s^*,b)\leq G(s^*,b^*)\leq
G(s,b^*)
\end{equation}
where $s^*$ and $b^*$ are the same as in (\ref{2.9}).
\end{prop}
The above statement is, actually, a part of Theorem \ref{thm2.1} (see (2.12))
but since
it holds true in a wider setting we give it separately. Observe also that the
above result is correct for different definitions of strategies. For instance,
we could take $s_i$ to be dependent only on the last time $a_{i-1}$ but in
order to be consistent we provide the argument only for the
strategies set $\cS$. In fact, it is easy to see that in the proof we
just use the assumption  $\sigma_i,\tau_i \geq (\sigma_{i-1}
\wedge\tau_{i-1}+1)\wedge N.$ Before we pass to
the proof of Proposition \ref{prop3.1} we shall derive the following key lemma.

\begin{lem}\label{lem3.2}
For $s,b\in \cS$ set
\[
F(s^*,b)=\big((\sigma^*_1,...,\sigma^*_L),(\tau_1,...,\tau_L)\big)\,\,
\mbox{and}\,\, F(s,b^*)= \big((\sigma_1,...,\sigma_L),(\tau^*_1,...,\tau^*_L)
\big).
\]
For every $0\leq n\leq N$ put
\[
X^{(0)}_n=Y^{(0)}_n=V^{(0)}_n=0
\]
and for any $0\leq i\leq L$ define
\[
R^{(i)}(\sigma,\tau)=\mathbb{I}_{\sigma<\tau}X^{(i)}_{\sigma}+
\mathbb{I}_{\sigma\geq\tau}Y^{(i)}_{\tau}.
\]
Then
\begin{eqnarray}\label{3.1-}
&E(R^{(i-1)}(\sigma^*_{L-i+2},\tau_{L-i+2})+H^{(L-i+1)}(\sigma^*_{L-i+1},
\tau_{L-i+1}))\\
&\leq E(R^{(i)}(\sigma^*_{L-i+1},\tau_{L-i+1}))\quad\mbox{and}\nonumber
\end{eqnarray}
\begin{eqnarray}\label{3.2-}
&E(R^{(i-1)}(\sigma_{L-i+2},\tau^*_{L-i+2})+H^{(L-i+1)}(\sigma_{L-i+1},
\tau^*_{L-i+1}))\\
&\geq E(R^{(i)}(\sigma_{L-i+1},\tau^*_{L-i+1})).\nonumber
\end{eqnarray}
\end{lem}
\begin{proof}
We shall give only the proof of inequality (\ref{3.1-}) since
(\ref{3.2-}) can be proven in a similar way. Set $\eta_{i}=
(\sigma^*_{i}\wedge\tau_i+1)\wedge N$ then we obtain from the definition
that
\begin{eqnarray*}
&R^{(i)}(\sigma^*_{L-i+1},\tau_{L-i+1})=\mathbb{I}_{\{\sigma^*_{L-i+1}<
\tau_{L-i+1} \}}X^{(i)}_{\sigma^*_{L-i+1}\wedge\tau_{L-i+1}}+
\mathbb{I}_{\{\sigma^*_{L-i+1}\geq\tau_{L-i+1}
\}}\\
&\times Y^{(i)}_{\sigma^*_{L-i+1}\wedge\tau_{L-i+1}}=
\mathbb{I}_{\{\sigma^*_{L-i+1}<\tau_{L-i+1}
\}}\big(X_{L-i+1}(\sigma^*_{L-i+1}\wedge\tau_{L-i+1})\\
&+E(V^{(i-1)}_{\eta_{L-i+1}}|\cF_{\sigma^*_{L-i+1}\wedge\tau_{L-i+1}}) \big)
+\mathbb{I}_{\{\sigma^*_{L-i+1}\geq \tau_{L-i+1} \}}\\
&\times\big(Y_{L-i+1}(\sigma^*_{L-i+1}\wedge\tau_{L-i+1})
+E(V^{(i-1)}_{\eta_{L-i+1}}|\cF_{\sigma^*_{L-i+1}\wedge\tau_{L-i+1}}) \big)\\
&=H^{(L-i+1)}(\sigma^*_{L-i+1},\tau_{L-i+1})+E(V^{(i-1)}_{\eta_{L-i+1}}
|\cF_{\sigma^*_{L-i+1}\wedge\tau_{L-i+1}}),
\end{eqnarray*}
 and so
\begin{eqnarray}\label{3.3-}
&E(R^{(i)}(\sigma^*_{L-i+1},\tau_{L-i+1}))\\
&=E(H^{(L-i+1)}(\sigma^*_{L-i+1},\tau_{L-i+1}))+E(V^{(i-1)}_{
\eta_{L-i+1}}).\nonumber
\end{eqnarray}
On the other hand,
\begin{eqnarray*}
&R^{(i-1)}(\sigma^*_{L-i+2},\tau_{L-i+2})=\mathbb{I}_{\{\sigma^*_{L-i+2}<
\tau_{L-i+2} \}}X^{(i-1)}_{\sigma^*_{L-i+2}}+\mathbb{I}_{\{\sigma^*_{L-i+2}
\geq\tau_{L-i+2} \}}Y^{(i-1)}_{\tau_{L-i+2}}\\
&\leq \mathbb{I}_{\{\sigma^*_{L-i+2}<\tau_{L-i+2} \}}V^{(i-1)}_{
\sigma^*_{L-i+2}}+\mathbb{I}_{\{\sigma^*_{L-i+2}\geq\tau_{L-i+2} \}}
V^{(i-1)}_{\tau_{L-i+2}}=V^{(i-1)}_{\sigma^*_{L-i+2}\wedge \tau_{L-i+2}}
\end{eqnarray*}
which holds true by the definition of $\sigma^*_{L-i+2}$ and the
fact that $Y^{(i)}_n\leq V^{(i)}_n$ for every $0\leq n\leq N$ and
$1\leq i\leq L$. Applying the last inequality in Lemma \ref{lem3.1} 
with $\theta=\eta_{L-i+1}$ we obtain that 
\begin{equation}\label{3.4-}
E\big (R^{(i-1)}(\sigma^*_{L-i+2},\tau_{L-i+2})\big )
\leq E(V^{(i-1)}_{\eta_{L-i+1}}).
\end{equation}
Now (\ref{3.1-}) follows from  (\ref{3.3-}) and (\ref{3.4-}).
\end{proof}
Observe that in the special case $s=s^*$ and $b=b^*$ if
\[
\big((\sigma^*_1,...,\sigma^*_L),(\tau^*_1,...,\tau^*_L)\big)=F(s^*,b^*)
\]
then inequalities (\ref{3.1-}) and (\ref{3.2-}) become equalities and
\begin{eqnarray}\label{3.5-}
&E(R^{(i-1)}(\sigma^*_{L-i+2},\tau^*_{L-i+2})+H^{(L-i+1)}(\sigma^*_{L-i+1},
\tau^*_{L-i+1}))\\
&=E(R^{(i)}(\sigma^*_{L-i+1},\tau^*_{L-i+1}))\nonumber
\end{eqnarray}
for every $1<i\leq L$.

\begin{proof}[Proof of Proposition \ref{prop3.1}]
For $b\in \cS$ let
\[
F(s^*,b)=\big((\sigma_1(s^*,b),...,\sigma_L(s^*,b)),(\tau_1(s^*,b),...,
\tau_L(s^*,b))\big)\,\,\,\mbox{and}
\]
\[
F(s^*,b^*)=
\big((\sigma_1(s^*,b^*),...,\sigma_L(s^*,b^*)),(\tau_1(s^*,b^*),
..,\tau_L(s^*,b^*))\big).
\]
We shall prove only the left hand side of (\ref{3.0}) while its right hand side
follows in the same way.
By Lemma \ref{lem3.2} we see that for every $1<i\leq L$,
\begin{eqnarray*}
&E(R^{(i-1)}(\sigma_{L-i+2}(s^*,b),\tau_{L-i+2}(s^*,b))+\sum_{j=1}^{L-i+1}
H^{(j)}(\sigma_{j}(s^*,b),\tau_{j}(s^*,b)))\\
&\leq
E(R^{(i)}(\sigma_{L-i+1}(s^*,b),\tau_{L-i+1}(s^*,b))+\sum_{j=1}^{L-i}
H^{(j)}(\sigma_{j}(s^*,b),\tau_{j}(s^*,b)))
\end{eqnarray*}
and for $(s^*,b^*)$,
\begin{eqnarray*}
&E(R^{(i-1)}(\sigma_{L-i+2}(s^*,b^*),\tau_{L-i+2}(s^*,b^*))+\sum_{j=1}^{L-i+1}
H^{(j)}(\sigma_{j}(s^*,b^*),\tau_{j}(s^*,b^*)))\\
&=
E(R^{(i)}(\sigma_{L-i+1}(s^*,b^*),\tau_{L-i+1}(s^*,b^*))+\sum_{j=1}^{L-i}
H^{(j)}(\sigma_{j}(s^*,b^*),\tau_{j}(s^*,b^*))).
\end{eqnarray*}
By induction it follows that
\begin{equation}
G(s^*,b)=E(\sum_{j=1}^L H^{(j)}(\sigma_j(s^*,b),\tau_j(s^*,b))) \leq
E(R^{(L)}(\sigma_1(s^*,b),\tau_1(s^*,b)))
\end{equation}
and  for $(s^*,b^*)$,
\begin{equation}G(s^*,b^*)=E(R^{(L)}(\sigma_1(s^*,b^*),\tau_1(s^*,b^*)))=V^{(L)}
_{0}
\end{equation}
where the last term is the value of the usual (one stopping) Dynkin game .
Observe that from the definition of $s^*,b^*$ for every $b\in \cS$
the inequality
\begin{equation}E(R^{(L)}(\sigma_1(s^*,b),\tau_1(s^*,b)))\leq E(R^{(L)}(
\sigma_1(s^*,b^*),\tau_1(s^*,b^*)))=V^{(L)}_0
\end{equation}
is just the saddle point property of the usual Dynkin game. From
(3.6), (3.7) and (3.8) it follows that
\begin{eqnarray*}
&G(s^*,b)\leq E(R^{(L)}(\sigma_1(s^*,b),\tau_1(s^*,b)))\\
&\leq
E(R^{(L)}(\sigma_1(s^*,b^*),\tau_1(s^*,b^*)))=G(s^*,b^*)=V^{(L)}_0.
\end{eqnarray*}
\end{proof}
As a consequence we obtain
\begin{cor}\label{cor3.1}
The multi stopping Dynkin game possess a saddle point $<s^*,b^*>$, and
so it has a value which is equal to $G(s^*,b^*)$.
\end{cor}
In the remaining part of this section we derive auxiliary lemmas
which will be used for the proof of Theorem \ref{thm2.2}.
\begin{dfn}\label{dfn3.1}
A function $\psi:{\mathbb{R}_{+}}\rightarrow{\mathbb{R}_{+}}$ is
 a piecewise linear function vanishing at $\infty$ if there exists a natural
 number $n$, such that
\begin{equation}\label{3.1}
\psi(y)=\sum_{i=1}^n \mathbb{I}_{[a_i,b_i)}(c_iy+d_i)
\end{equation}
where $c_1,...,c_n,d_1,...,d_n\in{\mathbb{R}}$ and
${\{[a_i,b_i)\}}_{i=1}^n$ is a sequence of disjoint finite
intervals.
\end{dfn}
\begin{lem}\label{lem3.3}
Let $A\geq{0}$ and
$\psi_1,\psi_2:{\mathbb{R}_{+}}\rightarrow{\mathbb{R}_{+}}$ be
continuous, decreasing and piecewise linear functions vanishing at
$\infty$. Define $\psi:{\mathbb{R}_{+}}\rightarrow{\mathbb{R}_{+}}$
and $\psi_A:{\mathbb{R}}\rightarrow{\mathbb{R}_{+}}$ by
\begin{eqnarray*}
&\psi(y)=\min_{\lambda\in{K(y)}}\big (p\psi_1(y+b\lambda)+(1-p)\psi_2(y+
a\lambda)\big )\nonumber\\
&\mbox{and} \ \ \psi_A(y)=\inf_{z\geq{(A-y)^{+}}}\big
(z+\psi(y+z-A)\big ).
\end{eqnarray*}
Then $\psi$ and $\psi_A$ are continuous, decreasing and piecewise
linear functions vanishing at $\infty$. Furthermore, there exists
$u\geq{(A-y)}^{+}$ such that
\begin{equation}\label{3.1+}
\psi_A(y)=u+\psi(y+u-A).
\end{equation}
\end{lem}
\begin{proof}
From Lemma 3.3 in \cite{DK} it follows that $\psi(y)$ is a
decreasing continuous function. Let us show that $\psi(y)$ is a
piecewise linear function vanishing at $\infty$. Since $0\in{K(y)}$
then
\begin{equation}\label{3.1++}
\psi(y)\leq p\psi_1(y)+(1-p)\psi_2(y)\leq \max(\psi_1(y),\psi_2(y)).
\end{equation}
There exists a natural number $n$ such that
\begin{equation}\label{3.2}
\psi_i(y)=\sum_{j=1}^n \mathbb{I}_{[a_j,b_j)}(c^{(i)}_jy+d^{(i)}_j),
\ i=1,2
\end{equation}
where $c^{(i)}_j,d^{(i)}_j\in{\mathbb{R}}$ and
${\{[a_i,b_i)\}}_{i=1}^n$ is a sequence of disjoint finite
intervals. Fix $y$ and define the function
$\phi_y(\lambda)=p\psi_1(y+b\lambda)+(1-p)\psi_2(y+a\lambda)$. From
(\ref{3.2}) it follows that there exists
\begin{equation}\label{3.2+}
\lambda\in{{\{-\frac{y}{b},-\frac{y}{a}\}}\cup
{\{{\frac{a_j-y}{b},\frac{b_j-y}{b},\frac{a_j-y}{a}
\frac{b_j-y}{b}\}}_{j=1}^n}}.
\end{equation}
such that $\psi(y)=\phi_y(\lambda)$. Thus, there
exists a finite sequence of real numbers $u_1,...,u_m,v_1,...,v_m$
such that for any $y$,
\begin{equation}\label{3.2++}
\psi(y)=u_iy+v_i
\end{equation}
for some $i$ (which depends on $y$). This together with
(\ref{3.1++}) and the fact that $\psi(y)$ is a continuous function
gives that $\psi(y)$ is a piecewise linear function vanishing at
$\infty$. Next, we deal with $\psi_A(y)$. Observe that
$\psi_A(y)\leq\psi(0)+(A-y)^{+}$. Thus
\begin{equation}\label{3.2+++}
\psi_A(y)=\inf_{(A-y)^{+} \leq z \leq
(A-y)^{+}+\psi(0)}\big (z+\psi(y+z-A)\big )
\end{equation}
and (\ref{3.1+}) follows from the fact that $\psi$ is continuous.
Choose $y_1<y_2$. Since $\psi(y)$ is a decreasing function then
\begin{equation}\label{3.3}
\psi_A(y_2)\leq \inf_{z\geq{(A-y_1)^{+}}}\big (z+\psi(y_2+z-A)\big )\leq
\inf_{z\geq{(A-y_1)^{+}}}\big (z+\psi(y_1+z-A)\big )=\psi_A(y_1).
\end{equation}
Thus $\psi_A(y)$ is a decreasing function. Now we want to prove
continuity. Choose $\epsilon>0$. Since $\psi(y)$ is a continuous
piecewise linear function vanishing at $\infty$ then there exists a
$\delta_1>0$ such that
\begin{equation}\label{3.4}
|y_1-y_2|<\delta_1\Rightarrow|\psi(y_1)-\psi(y_2)| <\epsilon.
\end{equation}
Set $\delta=\min(\epsilon,\delta_1)$. We will show that
\begin{equation}\label{3.4+}
|y_1-y_2|<\frac{\delta}{2}\Rightarrow|\psi_A(y_1)-\psi_A(y_2)|\leq
2\epsilon
\end{equation}
assuming without loss of generality that $y_1<y_2$. There
exists $u\geq{(A-y_2)^{+}}$ such that
\begin{equation}\label{3.4++}
\psi_A(y_2)=u+\psi(y_2+u-A).
\end{equation}
If $u\geq{(A-y_1)^{+}}$ then using (\ref{3.4},)
\begin{equation}\label{3.5}
\psi_A(y_1)-\psi_A(y_2)\leq u+\psi(y_1+u-A)-(u+\psi(y_2+u-A))\leq
\epsilon.
\end{equation}
If $u<{(A-y_1)^{+}}$ then $|u-(A-y_1)^{+}|\leq
(A-y_1)^{+}-(A-y_2)^{+}\leq\frac{\delta}{2}$ and
$|(y_1+(A-y_1)^{+}-A)-(y_2+u-A)|\leq\delta$. Thus from (\ref{3.4})
it follows that
\begin{equation}\label{3.6}
\psi_A(y_1)-\psi_A(y_2)\leq
(A-y_1)^{+}+\psi(y_1+(A-y_1)^{+}-A)-(u+\psi(y_2+u-A))\leq 2\epsilon.
\end{equation}
By (\ref{3.5}) and (\ref{3.6}) we obtain (\ref{3.4+}) and conclude
that $\psi_A(y)$ is a continuous function. Next, let
\begin{equation}\label{3.7}
\psi_(y)=\sum_{i=1}^k \mathbb{I}_{[\alpha_i,\beta_i)}(w_iy+x_i)
\end{equation}
where $k$ is a natural number, $w_i,x_i\in\mathbb{R}$ and
${\{[\alpha_i,\beta_i)}\}_{i=1}^k$ is a sequence of disjoint finite
intervals. Fix $y$ and define the function
$\phi_{A,y}(z)=z+\psi(y+z-A)$. From (\ref{3.2+++}) and (\ref{3.7})
it follows that there exists
\begin{equation*}
z\in{{\{(A-y)^{+}\}}\cup {\{{\alpha_i+A-y,\beta_i+A-y\}}_{i=1}^k}}
\end{equation*}
such that $\psi_A(y)=\phi_{A,y}(z)$. Hence, as before we see that
there exists a finite sequence of real numbers
$U_1,...,U_M,V_1,...,V_M$ such that for any $y$,
\begin{equation*}
\psi_A(y)=U_iy+V_i
\end{equation*}
for some $i$ which depends on $y$. This together with (\ref{3.2+++})
and the fact that $\psi_A(y)$ is a continuous function gives that
$\psi_A(y)$ is a piecewise linear function vanishing at $\infty$.
\end{proof}
\begin{lem}\label{lem3.4}
For any ${0}\leq{k}\leq{N}$ and $0\leq{j}\leq{L}$ and $u_1,...,u_k\in
\{a,b\}$ the function $J_k(\cdot,j,u_1,...,u_k)$ is continuous, decreasing,
piecewise linear and vanishing at $\infty$.
\end{lem}
\begin{proof}
We will use backward induction in $k$. For $k=N$ the statement
follows from (\ref{2.15}). Suppose the statement is correct for
$k=n+1$ and prove it for $k=n$. Fix $j>0$ (for $j=0$ the statement
is clear) and $u_1,...,u_n\in \{a,b\}$. Set
$\psi^{(i)}_1(y)=J_{n+1}(y,i,u_1,...,u_n,b)$ and $\psi^{(i)}_2(y)=
J_{n+1}(y,i,u_1,...,u_n,a)$. From the induction hypothesis it
follows that $\psi^{(i)}_1,\psi^{(i)}_2$ are continuous, decreasing
and piecewise linear functions vanishing at $\infty$. Thus, applying
Lemma \ref{lem3.1} to the functions
$\psi^{(j-1)}_1(y),\psi^{(j-1)}_2(y)$ and
$A=g^{(L-j+1)}_n(u_1,...,u_n)$ we obtain that the first term in
(\ref{2.16}) is a continuous, decreasing and piecewise linear
function vanishing at $\infty$ (with respect to $y$). Similarly we
obtain that the second term in (\ref{2.16}) is a continuous,
decreasing and a piecewise linear function vanishing at $\infty$.
Using Lemma \ref{lem3.1} for the functions
$\psi^{(j)}_1(y),\psi^{(j)}_2(y)$ we see that the third term in
(\ref{2.16}) is a continuous, decreasing and a piecewise linear
function vanishing at $\infty$. Thus $J_n(\cdot,j,u_1,...,u_n)$ is a
continuous, decreasing and piecewise linear function vanishing at
$\infty$ completing the proof.
\end{proof}

\section{Hedging and fair price}
In this section we prove Theorem \ref{thm2.1} starting with the following
observation.

\begin{lem}\label{lem4.1-}
Assume $Y_k,V_{k+1}$ are random variables which are respectively
$\cF_k$ and $\cF_{k+1}$ measurable. Assume that $Y_k\geq \tilde
E(V_{k+1}|\cF_{k})$. Then there exist a $\cF_k$-measurable random
variable $\gamma_k$ such that
\begin{equation}\label{4.0}
Y_k+\gamma_{k}(S_{k}-S_{k+1})\geq V_{k+1}.
\end{equation}
\end{lem}
\begin{proof} Set $V_k=\tilde E(V_{k+1}|\cF_{k})$. Then by the martingale
representation theorem in the binomial model (see, for instance \cite{SKKM})
there exists a $\cF_k$-measurable random variable $\gamma_k$ such that
\[
V_{k+1}=V_k+\gamma_{k}(S_{k}-S_{k+1})
\]
and (\ref{4.0}) follows.
\end{proof}

Next, we define a special portfolio strategy $\pi^*=(x^*,\gamma^*)$
setting \newline
$x^*=G(s^*,b^*)= V^{(L)}_0$ and taking $\gamma^*(k,i,y)$ to be
the random variable $\gamma_k$ from Lemma \ref{lem4.1-} with respect to
$Y_k=y$ and $V_{k+1}=V^{(L-i+1)}_{k+1}\mathbb{I}_{\{y\geq \tilde
E(V^{(L-i+1)}_{k+1}|\cF_k)\}}$. Note that if $y\geq \tilde
E(V^{(L-i+1)}_{k+1}|\cF_k)$ then by Lemma \ref{lem4.1-},
\begin{equation}\label{4.1-}
y+\gamma^*(k,i,y)(S_{k+1}-S_k)\geq V^{(L-i+1)}_{k+1}.
\end{equation}
Now we obtain.
\begin{lem}\label{lem4.2-}
The pair $(\pi^*,s^*)$ is a perfect hedge.
\end{lem}
\begin{proof}
Let $b\in \cS$ be any stopping strategy. Set
$F(s^*,b)=((\sigma_1,...,\sigma_L),(\tau_1,...,\tau_L)).$ In order to
derive that the pair $(\pi^*,S^*)$ is a perfect hedge we have to show
that for every $0\leq k\leq N$,
\[
V^{(\pi^*,s^*,b)}_k\geq 0.
\]
In fact, we shall see that for every $0\leq k\leq L$,
\begin{equation}\label{4.2-}
V^{(\pi^*,s^*,b)}_k\geq \tilde E(V^{(L-c_k)}_{k+1}|\cF_{k})
\end{equation}
where $c_k=\sum_{i=1}^{L}\mathbb{I}_{\{\sigma_i\wedge\tau_i\leq
k\}}$. Since $c_k$ is measurable with respect to the
$\sigma$-algebra $\cF_k$ the inequality (\ref{4.2-}) is a
consequence of the following inequalties
\[
V^{(\pi^*,s^*,b)}_k\mathbb{I}_{c_k=i}\geq \tilde E(V^{(L-i)}_{k+1}|\cF_k)
\mathbb{I}_{c_k=i},\,\,\,\,  1\leq i\leq L.
\]
For every $1\leq i\leq L$ the above inequality will be proved by
induction in $k$. For $k=0$  we may have either $c_0=0$ or
 $c_0=1$ where the second event occurs when either the writer or the
holder exercised the first claim at the time $k=0$. If $c_0=0$ then
by (\ref{2.6}),
\[
V^{(\pi^*,s^*,b)}_0=x^*=V^{(L)}_0.
\]
Since $V^{(L)}_{\sigma^*_1\wedge k}$ is a supermartingale
with respect to $\{\cF_{k}\}_{k=0}^N$ and
$1\leq \sigma^*_1\wedge\tau_1 \leq \sigma^*_1$ it follows that on the event $c_0=0$, which is $\cF_0$ measurable, we have

\[
V^{(L)}_0\geq \tilde E(V^{(L)}_{\sigma^*_1\wedge 1})=\tilde E(V^{(L)}_1).
\]
If $c_0=1$ we obtain
\begin{eqnarray*}
&V^{(\pi^*,s^*,b)}_0=V^{(L)}_0-H^{(1)}(\sigma^*_1,\tau_1)
\geq\big(X^{(L)}_{\sigma^*_1}-X_1(\sigma^*_1)\big)\mathbb{I}_{\{
\sigma^*_1<\tau_1 \}}\\
&+\big(Y^{(L)}_{\tau_1}-Y_1(\tau_1)\big)\mathbb{I}_{
\{\sigma^*_1\geq\tau_1 \}}
=\tilde E(V^{(L-1)}_{\sigma^*_1\wedge\tau_1+1}|\cF_{\sigma^*_1\wedge\tau_1})
=\tilde E(V^{(L-1)}_{1})
\end{eqnarray*}
where the first equality is (\ref{2.6}), the inequality is derived from
the definition of the stopping time $\sigma^*_1$ and the fact that
$V^{(L)}\geq Y^{(L)}$ and the last equalities follow from the
definitions of $X^{(L)}$ and $Y^{(L)}$ and the fact that
$\sigma^*_1\wedge\tau_1=0$ when $c_0=1$.

Next, let $0<k\leq N$. Assume, first, that $c_k=i<L$. Then by
the definition of $c_k$ it follows that $\sigma^*_i\wedge \tau_i\leq
k$. Similarly to the case $k=0$ we may have either
$\sigma^*_i\wedge\tau_i<k$ or $\sigma^*_i\wedge\tau_i=k$.
If $\sigma^*_i\wedge\tau_i<k$ then $c_{k-1}=i$ and so by (\ref{2.6}),
\[
V^{(\pi^*,s^*,b)}_k=V^{(\pi^*,s^*,b)}_{k-1}+\gamma^*(k-1,i+1,
V^{(\pi^*,s^*,b)}_{k-1})(S_{k}-S_{k-1}).
\]
where the equality holds on the $\cF_k$ event $\sigma^*_i\wedge\tau_i<k.$
By the induction hypothesis we obtain on this event that
\[
V^{(\pi^*,s^*,b)}_{k-1}\geq \tilde E(V^{(L-i)}_{k}|\cF_{k-1}).
\]
By (\ref{4.1-}) it follows that
\[
V^{(\pi^*,s^*,b)}_k=V^{(\pi^*,s^*,b)}_{k-1}+\gamma^*(k-1,i+1,
V^{\pi^*,s^*,b}_{k-1})(S_{k}-S_{k-1})\geq V^{(L-i)}_k.
\]
Since $c_k=i$ the definition of $c_k$ yields that
$\sigma^*_{i+1}\geq \sigma^*_{i+1}\wedge\tau_{i+1} \geq k+1$, and so
from the supermartingale property of
$V^{(L-i)}_{\sigma^*_{i+1}\wedge l }$ for $l\geq k+1 \geq
\sigma^*_{i}\wedge\tau^*_i+1$ we obtain
\[
V^{(L-i)}_k\geq \tilde E(V^{(L-i)}_{\sigma^*_{i+1}\wedge
k+1}|\cF_k)=\tilde E(V^{(L-i)}_{ k+1}|\cF_k).
\]
Now consider the $\cF_{k}$ event $\sigma^*_i\wedge\tau_i=k$. Then $c_{k-1}=i-1$ and (\ref{2.6}) becomes
\[
V^{(\pi^*,s^*,b)}_k=V^{(\pi^*,s^*,b)}_{k-1}+\gamma^*(k-1,i,
V^{(\pi^*,s^*,b)}_{k-1})(S_{k}-S_{k-1})-H(\sigma^*_i,\tau_i).
\]
Since $c_{k-1}=i-1$ the induction hypothesis yields that
\[
V^{(\pi^*,s^*,b)}_{k-1}\geq \tilde E(V^{(L-i+1)}_{k}|\cF_{k-1}),
\]
and so from the definition of $\gamma^*(k-1,i,y)$ we obtain that
\begin{eqnarray*}
&V^{(\pi^*,s^*,b)}_{k-1}+\gamma^*(k-1,i,V^{(\pi^*,s^*,b)}_{k-1})(S_{k}-S_{k-1})-
H^{(i)}(\sigma^*_i,\tau_i)\\
& \geq V^{(L-i+1)}_{k}-H^{(i)}(\sigma^*_i,\tau_i)
=V^{(L-i+1)}_{\sigma^*_i\wedge \tau_i}-H^{(i)}(\sigma^*_i,\tau_i).
\end{eqnarray*}
From the definition of $\sigma^*_i$, the fact that $V^{(i)}\geq
Y^{(i)}$ and the definition of $X^{(i)},Y^{(i)}$ it follows that
\begin{eqnarray*}
&V^{(L-i+1)}_{\sigma^*_i\wedge \tau_i}-H^{(i)}(\sigma^*_i,\tau_i)
\geq \big(X^{(L-i+1)}_{\sigma^*_i}-X_i(\sigma^*_i) \big)\mathbb{I}_
{\{\sigma^*_i<\tau_i\}}\\
&+\big(Y^{(L-i+1)}_{\tau_i}-Y_i(\tau_i)
\big)\mathbb{I}_{\{\sigma^*_i
\geq\tau_i\}}\\
&=\tilde E(V^{(L-i)}_{\sigma^*_i\wedge\tau_i+1}|\cF_{\sigma^*_i
\wedge\tau_i})= \tilde E(V^{(L-i)}_{k+1}|\cF_{k}).
\end{eqnarray*}
We are left only with the event $c_k=L$. On this event the inequality
(\ref{4.2-}) is reduced to
\[
V^{(\pi^*,s^*,b)}_k\geq 0.
\]
If $\sigma^*_L\wedge\tau_L=k$ then the proof is the same as above in the
case $\sigma^*_i\wedge\tau_i=k$ for $i<L$. In the case
$\sigma^*_L\wedge\tau_L <k$ there are no claims left to
exercise or cancel, and so by the definition of $\gamma^*$ we see
that the portfolio value will stay nonnegative till the time $N$.
\end{proof}
Next, we show that $x^*=V^{(L)}_0$ is the minimal initial capital for
a perfect hedge.
\begin{lem}\label{lem4.3-}
Assume that the pair $(\pi,s)=((x,\gamma),s)$ is a perfect hedge.
Then
\[
x\geq x^*=V^{(L)}_0.
\]
\end{lem}
\begin{proof}
Let $b^*$ be the stopping strategy for the buyer defined in
(\ref{2.9}) and set
\[
F(s,b^*)=((\sigma_1,...,\sigma_L),(\tau^*_1,...,\tau^*_L)).
\]
 We want to show that
\begin{equation}\label{4.3-}
V^{(\pi,s,b^*)}_k\geq \tilde E(V^{(L-c_k)}_{(k+1) \wedge{N}}|\cF_k)
\end{equation}
where $c_k$ is computed with respect to $(s,b^*)$. Recall that for
every $0\leq k\leq N$ the function $c_k$ is $\cF_k$ measurable and
since inequality (\ref{4.3-}) is between $\cF_k$ measurable
functions we can prove $(\ref{4.3-})$ separately on the events
$c_k=i.$

The inequality (\ref{4.3-}) will be proved by the backward induction
in $k$. When $c_k=L$ the right hand side of (\ref{4.3-}) is zero and
the definition of a perfect hedge yields that the left hand side of
(\ref{4.3-}) is non negative, hence (\ref{4.3-}) is true in these
cases. Next, assume that $c_k=i$ where $0\leq i\leq L-1$ (thus
$k<N$). We
 split the proof into two events $c_{k+1}=i$ and $c_{k+1}=i+1$. In the
 second event the $(i+1)$-th claim was exercised or canceled at the time $k+1$.

We begin with the event $c_{k+1}=i$ (thus $k<N-1$). From the
induction hypothesis it follows that
\[
V^{(\pi,s,b^*)}_{k+1}\geq \tilde E(V^{(L-i)}_{k+2}|\cF_{k+1})=
\tilde E(V^{(L-i)}_{\tau^*_{i+1}\wedge(k+2)}|\cF_{k+1})\geq
V^{(L-i)}_{\tau^*_{i+1} \wedge(k+1)}.
\]
The equality here holds true since $\tau^*_{i+1}\geq
\tau^*_{i+1}\wedge\sigma_{i+1}\geq k+2>k+1>\sigma_i\wedge\tau^*_i$ when
 $c_{k+1}=c_k=i$ and the last inequality follows from the submartingale
 property of $V^{(L-i)}_{\tau^*_{i+1}\wedge l}$ for $l>\sigma_i\wedge\tau^*_i.$
Since $c_i=k$ we have from (\ref{2.6}) that
\[
V^{\pi,s,b^*}_{k+1}=V^{(\pi,s,b^*)}_k+\gamma(k,i+1,V^{(\pi,s,b^*)}_k)(S_{k+1}-
S_{k}).
\]
Since $S_k$ is a martingale with respect to $\tilde P$ then
using this equality and taking the conditional expectation with
respect to $\cF_{k}$ in the above inequality we obtain
\[
V^{(\pi,s,b^*)}_{k}\geq \tilde E(V^{(L-i)}_{\tau^*_{i+1}\wedge(k+1)}|\cF_k)=
\tilde E(V^{(L-i)}_{k+1}|\cF_k).
\]
Next, assume that $c_{k+1}=i+1$ which together with the assumption
$c_k=i$ yields that $\sigma_{i+1}\wedge\tau^*_{i+1}=k+1$. By the
induction hypothesis it follows that
\[
V^{(\pi,s,b^*)}_{k+1}\geq \tilde E(V^{(L-i-1)}_{k+2}|\cF_{k+1})=\tilde
E(V^{(L-i-1)}_{\sigma_i\wedge\tau^*_i+1}|\cF_{\sigma_i\wedge\tau^*_i}),
\]
and so
\begin{eqnarray*}
&V^{(\pi,s,b^*)}_{k+1}+H^{(i+1)}(\sigma_{i+1},\tau^*_{i+1})\geq \tilde
E(V^{(L-i-1)}_{\sigma_i\wedge\tau^*_i+1}|\cF_{\sigma_i\wedge\tau^*_i})+
H^{(i+1)}(\sigma_{i+1},\tau^*_{i+1})\\
&=X^{(L-i)}_{\sigma_{i+1}}\mathbb{I}_{\{\sigma_{i+1}<\tau^*_{i+1}\}}+
Y^{(L-i)}_{\tau^*_{i+1}}\mathbb{I}_{\{\sigma_{i+1}\geq\tau^*_{i+1}\}}
\geq V^{(L-i)}_{\sigma_{i+1}\wedge\tau^*_{i+1}}= V^{(L-i)}_{k+1}
\end{eqnarray*}
where the second inequality holds true since $X^{(L-i)}\geq V^{(L-i)}$
 and in view of the definition of the stopping time $\tau^*_{i+1}$.
 On the event $c_k=i$ and $c_{k+1}=i+1$ the equality (\ref{2.6}) becomes
\[
V^{(\pi,s,b^*)}_{k+1}+H^{(i+1)}(\sigma_{i+1},\tau^*_{i+1})=
V^{(\pi,s,b^*)}_{k}+\gamma(k,i+1,V^{(\pi,s,b^*)}_{k})(S_{k+1}-S_{k})
\]
and taking the conditional expectation of the above inequality
with respect to the sigma algebra $\cF_k$ we obtain that
\[
V^{(\pi,s,b^*)}_{k}\geq \tilde E(V^{(L-i)}_{k+1}|\cF_k)
\]
completing the proof of (\ref{4.3-}). As a special case of (\ref{4.3-})
for $k=0$ it follows that
\[
V^{(\pi,s,b^*)}_0\geq \tilde E(V^{(L-c_0)}_1).
\]
If $c_0=0$ then $\tau^*_1\geq \sigma_1\wedge\tau^*_1\geq 1$ and since
$V^{(L)}_{\tau^*_1\wedge l},\, l\geq 0$ is a submartingale we see that
\[
V^{(\pi,s,b^*)}_0\geq \tilde E(V^{(L-c_0)}_1)=\tilde E(V^{(L)}_{\tau^*_1
\wedge 1})\geq V^{(L)}_{0}=x^*.
\]
If $c_0=1$ then $\sigma_1\wedge\tau^*_1=0$, and so
\[
x-H(\sigma_1,\tau^*_1)=V^{(\pi,s,b^*)}_0\geq \tilde E(V^{(L-1)}_1)
\]
which can also be written in the form
\begin{eqnarray*}
&x\geq \tilde E(V^{(L-1)}_{\sigma_1\wedge\tau^*_1+1})+H(\sigma_1,\tau^*_1)\\
&=X^{(L)}_{\sigma_1}\mathbb{I}_{\{\sigma_1<\tau^*_1\}}+Y^{(L)}_{\tau^*_1}
\mathbb{I}_{\{\sigma_1\geq\tau^*_1\}}\geq V^{(L)}_{\sigma_1\wedge\tau^*_1}=
V^{(L)}_0=x^*
\end{eqnarray*}
or in short
\[
x\geq x^*=V^{(L)}_0.
\]
\end{proof}
We can now prove the main theorem of this section.
\begin{proof}[Proof of Theorem \ref{thm2.1}]
From Lemma \ref{lem4.2-} and the definition of the fair price $V^*$ we obtain
that
\[
V^{(L)}_0=x^*\geq V^*.
\]
On the other hand, Lemma \ref{lem4.3-} yields that
\[
x^*\leq V^*.
\]
By Proposition 3.1,
\[
G(b,s^*)\leq G(b^*,s^*)=V^{(L)}_0\leq G(b^*,s)
\]
for any pair of stopping strategies $b,s\in \cS$ which gives (2.13)
 and collecting together the above inequalities we obtain (2.11).
 Since $\pi^*=(V^{(L)}_0,\gamma^*)$ we it follows that
$\pi^* \in \cA(V^{(L)}_0)$ and by Lemma \ref{lem3.2} the pair $(\pi^*,s^*)$
is a perfect hedge completing the proof of Theorem \ref{thm2.1}.
\end{proof}

\section{Shortfall risk and its hedging}\label{sec:5}\setcounter{equation}{0}
In this section we derive Theorem \ref{thm2.2} whose proof is quite technical
but the main idea is to apply Lemma \ref{lem3.1} to Dynkin's games with
appropriately constructed payoff processes which via Lemma \ref{lem4.0} below
enables us to produce a hedge for the shortfall risk whose optimality is
established by means of Lemmas \ref{lem4.1} and \ref{lem4.2} below.

For any $I\in\mathcal{I}$ set
\begin{eqnarray}\label{5.1}
&Z^{(I)}(y,k,j,u_1,...,u_k)=y-f^{(L-j+1)}_k(u_1,...,u_k)+I(k,L-j+1,\\
&y-f^{(L-j+1)}_k(u_1,...,u_k)) \ \mbox{and} \
\tilde{Z}^{(I)}(y,k,j,u_1,...,u_k)=y-\nonumber\\
&g^{(L-j+1)}_k(u_1,...,u_k)
+I(k,L-j+1,y-g^{(L-j+1)}_k(u_1,...,u_k)).\nonumber
\end{eqnarray}
Observe that if at the moment $k$ the seller pays his $(L-j+1)$-th
payoff and this is his first payoff at this moment (at $k=N$ more than
 one payoff can occur) then his portfolio value
after this payoff is either $Z^{(I)}(y,k,j,\rho_1,...,\rho_k)$
or $\tilde{Z}^{(I)}(y,k,j,\rho_1,...,\rho_k)$
in the case of an exercise or a cancellation,
respectively, provided an infusion of capital before the payoff is $y$
(where $\rho_i$ is the same as in (\ref{2.2})). Next, for any
$\pi=(x,\gamma)\in\mathcal{A}(x)$ and $I\in\mathcal{I}$ define
\begin{eqnarray}\label{5.2}
&U^{(\pi,I)}(y,k,j,u_1,...,u_{k+1})=
Z^{(I)}(y,k,j,u_1,...,u_k)+\mathbb{I}_{j>1}\\
&\times\gamma(k,L-j+2,Z^{(I)}(y,k,j,u_1,...,u_k))S_0u_{k+1}\prod_{i=1}^{k}
(1+u_i)\ \mbox{and} \nonumber\\
&\tilde{U}^{(\pi,I)}(y,k,j,u_1,...,u_{k+1})=
\tilde{Z}^{(I)}(y,k,j,u_1,...,u_k)+\mathbb{I}_{j>1}\nonumber\\
&\times\gamma(k,L-j+2,\tilde{Z}^{(I)}(y,k,j,u_1,...,u_k))S_0u_{k+1}
\prod_{i=1}^{k}(1+u_i).\nonumber
\end{eqnarray}
Note that if at the moment $k<N$ the seller pays his
$(L-j+1)$-th payoff then his portfolio value at the time $k+1$ before any
payoffs is either $U^{(\pi,I)}(y,k,j,\rho_1,...,\rho_{k+1})$ or
$\tilde{U}^{(\pi,I)}(y,k,j,\rho_1,...,\rho_{k+1})$ in the case of an exercise
or a cancellation, respectively, at the time $k$ provided that his portfolio
value before payoffs was $y$.  Finally, for
any $(\pi,I)\in\mathcal{A}\times\mathcal{I}$ define a sequence of
functions
$J^{(\pi,I)}_k:\mathbb{R}_{+}\times\{0,...,L\}\times{\{a,b}\}^k\rightarrow{
\mathbb{R}_{+}}$,
$0\leq{k}\leq{N}$ setting, first,
\begin{eqnarray}\label{5.3}
&J^{(\pi,I)}_N(y,j,u_1,...,u_N)=((\sum_{i=L-j+1}^L
f^{(i)}_N(u_1,...,u_N))-y)^{+}, \ j>0,\\
&J^{(\pi,I)}_k(y,0,u_1,...,u_k)=0, \ 0\leq{k}\leq{N}.\nonumber
\end{eqnarray}
Next, for $k<N$ and $j>0$ set
\begin{eqnarray}\label{5.4}
&J^{(\pi,I)}_k(y,j,u_1,...,u_k)=I(k,L-j+1,y-f^{(L-j+1)}_k(u_1,...,u_k))\\
&+pJ^{(\pi,I)}_{k+1}(U^{(\pi,I)}(y,k,j,u_1,...,u_k,b),j-1,u_1,...,u_k,b)
\nonumber \\
&+(1-p)J^{(\pi,I)}_{k+1}(U^{(\pi,I)}(y,k,j,u_1,...,u_k,a),j-1,u_1,...,u_k,a)
\nonumber
\end{eqnarray}
if
\begin{eqnarray}\label{5.5}
&I(k,L-j+1,y-f^{(L-j+1)}_k(u_1,...,u_k))\\
&+pJ^{(\pi,I)}_{k+1}(U^{(\pi,I)}(y,k,j,u_1,...,u_k,b),j-1,u_1,...,u_k,b)
\nonumber \\
&+(1-p)J^{(\pi,I)}_{k+1}(U^{(\pi,I)}(y,k,j,u_1,...,u_k,a),j-1,u_1,...,u_k,a)
\nonumber\\
&\geq I(k,L-j+1,y-g^{(L-j+1)}_k(u_1,...,u_k))\nonumber\\
&+pJ^{(\pi,I)}_{k+1}(\tilde{U}^{(\pi,I)}(y,k,j,u_1,...,u_k,b),j-1,u_1,...,u_k,b)
\nonumber \\
&+(1-p)J^{(\pi,I)}_{k+1}(\tilde{U}^{(\pi,I)}(y,k,j,u_1,...,u_k,a),j-1,u_1,...,
u_k,a)\nonumber
\end{eqnarray}
and
\begin{eqnarray}\label{5.6}
&J^{(\pi,I)}_k(y,j,u_1,...,u_k)=\min\Bigg(I(k,L-j+1,y-
g^{(L-j+1)}_k(u_1,...,u_k))\\
&+pJ^{(\pi,I)}_{k+1}(\tilde{U}^{(\pi,I)}(y,k,j,u_1,...,u_k,b),j-1,
u_1,...,u_k,b)\nonumber \\
&+(1-p)J^{(\pi,I)}_{k+1}(\tilde{U}^{(\pi,I)}(y,k,j,u_1,...,u_k,a),j-1,u_1,...,
u_k,a),\nonumber\\
&\max\Bigg(I(k,L-j+1,y-f^{(L-j+1)}_k(u_1,...,u_k))\nonumber\\
&+pJ^{(\pi,I)}_{k+1}(U^{(\pi,I)}(y,k,j,u_1,...,u_k,b),j-1,u_1,...,u_k,b)
\nonumber \\
&+(1-p)J^{(\pi,I)}_{k+1}(U^{(\pi,I)}(y,k,j,u_1,...,u_k,a),j-1,u_1,...,u_k,a),
\nonumber\\
&pJ^{(\pi,I)}_{k+1}(y+\gamma(k,L-j+1,y)S_0b\prod_{i=1}^{k}(1+u_i),j,u_1,...,u_k,b)
\nonumber \\
&+(1-p)J^{(\pi,I)}_{k+1}(y+\gamma(k,L-j+1,y)S_0a\prod_{i=1}^{k}(1+u_i),j,u_1,...,
u_k,a)\Bigg)\Bigg)\nonumber
\end{eqnarray}
if the inequality in (\ref{5.5}) does not hold true.

For any $j\geq{1}$ and $k\leq{N}$ consider the set
$\mathcal{S}^{(j)}_k$ of sequences $s=(s_1,...,s_j)$ such that
$s_1\in\Gamma_k$  and for $i>1$, $s_i:C_{i-1}\rightarrow{\Gamma}$ is
a map which satisfy
\[
s_i((a_1,...,a_{i-1}),(d_1,...,d_{i-1}))\in
\Gamma_{N\wedge{(1+a_{i-1})}}.
\]
 Next, define a map
$F:\mathcal{S}^{(j)}_k\times\mathcal{S}^{(j)}_k\rightarrow{\Gamma^{j}\times
\Gamma^{j}}$
by
\[
F(s,b)=((\sigma_1,...,\sigma_j),(\tau_1,...,\tau_j))
\]
 in the same way as in (\ref{2.4+}). Fix
$(\pi,I)\in\mathcal{A}\times{\mathcal{I}}$, $j\geq{m}\geq{1}$,
$k\leq{N}$ and $y\geq{0}$. Consider a swing option which starts at
the time $k$ where the initial capital of the seller equal
$y$, the number of remaining payoffs is $m$ and it starts from the
$(L-j+1)$-th claim. Let $z=(a,d)=((a_1,...,a_m),(d_1,...,d_m))\in C_m$
be a sequence which represents the history of the payoffs. Set
$c_n=c_n(z)=L-j+\sum_{i=1}^m \mathbb{I}_{a_i\leq{n}}$. Define the
stochastic processes ${\{W^{(y,\pi,I,k,j,z)}_n\}}_{n=k}^N$ and
${\{V^{(y,\pi,I,k,j,z)}_n\}}_{n=k}^N$ by
\begin{eqnarray}\label{5.7}
&W^{(y,\pi,I,k,j,z)}_k=y, \ V^{(y,\pi,I,k,j,z)}_k=W^{(y,\pi,I,k,j,z)}_k-
\mathbb{I}_{a_1=k}\\
&\times \bigg
(\mathbb{I}_{d_1=1}X_{L-j+1}(k)+\mathbb{I}_{d_1=0}Y_{L-j+1}(k)+\mathbb{I}_{k=N}\sum_{i=L-j+2}^L
Y_{i}(N)
-I(k,L-j+1,\nonumber\\
&W^{(y,\pi,I,k,j,z)}_k-\mathbb{I}_{d_1=1}X_{L-j+1}(k)-\mathbb{I}_{d_1=0}
Y_{L-j+1}(k))\bigg ) \quad\mbox{and for}\,\,\, n>k, \nonumber\\
&V^{(y,\pi,I,k,j,z)}_n=W^{(y,\pi,I,k,j,z)}_{n-1}+\mathbb{I}_{c_{n-1}<L}
\gamma(n-1,c_{n-1}+1,\nonumber\\
&W^{(y,\pi,I,k,j,z)}_{n-1})(S_n-S_{n-1}),
W^{(y,\pi,I,k,j,z)}_n=V^{(y,\pi,I,k,j,z)}_n-\mathbb{I}_{c_{n-1}<L}
\mathbb{I}_{a_{c_{n-1}+1}=n}\nonumber\\
&\times\bigg (X_{c_{n-1}+1}(n)\mathbb{I}_{d_{c_{n-1}+1}=1}+
Y_{c_{n-1}+1}(n)\mathbb{I}_{d_{c_{n-1}+1}=0}+\mathbb{I}_{n=N}
\sum_{i=c_{n-1}+2}^L
Y_{i}(N)\nonumber\\
&-I(n,c_{n-1}+1,V^{(y,\pi,I,k,j,z)}_n-X_{c_{n-1}+1}(n)
\mathbb{I}_{d_{c_{n-1}+1}=1}-
Y_{c_{n-1}+1}(n)\mathbb{I}_{d_{c_{n-1}+1}=0})\bigg ).\nonumber
\end{eqnarray}
Similarly to (\ref{2.11}) we conclude (under the conditions that
were described above) that if the contract was not exercised at a
moment $n$ then $W^{(y,\pi,I,k,j,z)}_n=V^{(y,\pi,I,k,j,z)}_n$ is the
portfolio value at this moment . If the contract was exercised at
the moment $n$ then $W^{(y,\pi,I,k,j,z)}_n$ and
$V^{(y,\pi,I,k,j,z)}_n$ are the portfolio values before and after
the payoff, respectively. For the case $m=0$ (no history of payoffs)
we define the stochastic processes
${\{W^{(y,\pi,I,k,j)}_n\}}_{n=k}^N$ by
\begin{eqnarray}\label{5.7+}
&W^{(y,\pi,I,k,j)}_k=y \ \mbox{and} \ \mbox{for} \ n>k,\\
&W^{(y,\pi,I,k,j)}_n=W^{(y,\pi,I,k,j)}_{n-1}+\gamma(n-1,L-j+1,
W^{(y,\pi,I,k,j)}_{n-1})(S_n-S_{n-1}).\nonumber
\end{eqnarray}
Clearly, $W^{(y,\pi,I,k,j)}_n$ is the portfolio value if no payoffs
were made until the moment $n$. Let $s\in\mathcal{S}^{(j)}_k$ and
$b\in\mathcal{S}^{(j)}_k$ be stopping strategies of the seller and
the buyer, respectively. Set
$((\sigma_1,...,\sigma_j),(\tau_1,...,\tau_j))=F(s,b)$,
$a_i=\sigma_i\wedge\tau_i$, $d_i=\mathbb{I}_{\sigma_i<\tau_i}$ and
$z=((a_1,...,a_i),(d_1,...,d_i))$. Define
\begin{eqnarray}\label{5.7++}
&W^{(y,\pi,I,k,j,s,b)}_n(\omega)=W^{(y,\pi,I,k,j,z(\omega))}_n(\omega)\,\,\,
\mbox{and}\,\,\,\\
&V^{(y,\pi,I,k,j,s,b)}_n(\omega)=V^{(y,\pi,I,k,j,z(\omega))}_n(\omega).
\nonumber
\end{eqnarray}
Similarly to (\ref{2.12}) the total infusion of capital is given by
\begin{equation}\label{5.8}
C(y,\pi,I,k,j,s,b)=\sum_{i=1}^{\alpha\wedge{j}}
I(\sigma_i\wedge\tau_i,i+L-j,W^{(y,\pi,I,k,j,s,b)}_{\sigma_i\wedge\tau_i}-
H^{(L-j+i)}(\sigma_i,\tau_i))
\end{equation}
where $\alpha=1+\sum_{i=1}^j \mathbb{I}_{\sigma_i\wedge\tau_i<N}$.
Thus for any $(\pi,I)\in\mathcal{A}\times{\mathcal{I}}$,
$j\geq{1}$, $k\leq{N}$, $s,b\in\mathcal{S}^{(j)}_k$ and $y\geq{0}$
we have the following definition for the shortfall risk
\begin{eqnarray}\label{5.9}
&R(y,\pi,I,k,j,s,b)=E(C(y,\pi,I,k,j,s,b)|\mathcal{F}_k), \\
&R(y,\pi,I,k,j,s)=\max_{b\in\mathcal{S}^{(j)}_k}R(y,\pi,I,k,j,s,b),\nonumber\\
&R(y,\pi,I,k,j)=\min_{s\in\mathcal{S}^{(j)}_k}R(y,\pi,I,k,j,s).\nonumber
\end{eqnarray}
Next, we define stopping strategies which will turn out to be
optimal. Let $(\pi,I)\in\mathcal{A}\times{\mathcal{I}}$, $j\geq{1}$,
$k\leq{N}$ and $y\geq{0}$. Define
$\tilde{s}(y,\pi,I,k,j)=(\tilde{s}_1,...,\tilde{s}_j)\in\mathcal{S}^{(j)}_k$
and
$\tilde{b}(y,\pi,I,k,j)=(\tilde{b}_1,...,\tilde{b}_j)\in\mathcal{S}^{(j)}_k$
by
\begin{eqnarray}\label{5.10}
&\tilde{s}_1=N\wedge\min{\Big\{n\geq{k}|J^{(\pi,I)}_n(W^{(y,\pi,I,k,j)}_n,j,
\rho_1,...,\rho_n)}\\
&\geq I(n,L-j+1,W^{(y,\pi,I,k,j)}_n-X_{L-j+i}(n))\nonumber\\
&+E(J^{(\pi,I)}_{n+1}(\tilde{U}^{(\pi,I)}(W^{(y,\pi,I,k,j)}_n,n,j,
\rho_1,...,\rho_{n+1}),j-1,\rho_1,...,\rho_{n+1})|\mathcal{F}_n)\Big\},\nonumber\\
&\tilde{b}_1=N\wedge\min\Big{\{n\geq{k}|J^{(\pi,I)}_n(W^{(y,\pi,I,k,j)}_n,j,
\rho_1,...,\rho_n)}\nonumber\\
&=I(n,L-j+1,W^{(y,\pi,I,k,j)}_n-Y_{L-j+1}(n))\nonumber\\
&+E(J^{(\pi,I)}_{n+1}(U^{(\pi,I)}(W^{(y,\pi,I,k,j)}_n,n,j,\rho_1,...,
\rho_{n+1}),j-1,\rho_1,...,\rho_{n+1})|\mathcal{F}_n)\Big\}.
\nonumber
\end{eqnarray}
For $i>1$ let $z=(a,d)=((a_1,...,a_{i-1}),(d_1,...,d_{i-1}))\in
C_{i-1}$ and define
\begin{eqnarray}\label{5.11}
&\tilde{s}_i(z)=N\wedge\min{\Big\{n>a_{i-1}|J^{(\pi,I)}_n(W^{(y,\pi,I,k,j,z)}_n,
j-i+1,\rho_1,...,\rho_n)}\\
&\geq I(n,L-j+i,W^{(y,\pi,I,k,j,z)}_n-X_{L-j+i}(n))\nonumber\\
&+E(J^{(\pi,I)}_{n+1}(\tilde{U}^{(\pi,I)}(W^{(y,\pi,I,k,j,z)}_n,n,j-i+1,\rho_1,
...,\rho_{n+1}),j-i,\rho_1,...,\rho_{n+1})|\mathcal{F}_n)\Big\},\nonumber\\
&\tilde{b}_i(z)=N\wedge\min{\Big\{n>a_{i-1}|J^{(\pi,I)}_n(W^{(y,\pi,I,k,j,z)}_n,
j-i+1,\rho_1,...,\rho_n)}\nonumber\\
&=I(n,L-j+i,W^{(y,\pi,I,k,j,z)}_n-Y_{L-j+i}(n))\nonumber\\
&+E(J^{(\pi,I)}_{n+1}(U^{(\pi,I)}(W^{(y,\pi,I,k,j,z)}_n,n,j-i+1,\rho_1,...,
\rho_{n+1}),j-i,
\rho_1,...,\rho_{n+1})|\mathcal{F}_n)\Big\}.\nonumber
\end{eqnarray}
The following two lemmas will be crucial for the proof of Theorem \ref{thm2.2}.

\begin{lem}\label{lem4.0}
Let $\pi,I\in\mathcal{A}\times\mathcal{I}$, $n\leq{N}$, $j\geq{1}$,
and $y\geq{0}$. Define the stochastic processes
${\{A_k\}}_{k=n}^N$ and ${\{D_k\}}_{k=n}^N$ by
\begin{eqnarray}\label{5.12}
&A_N=D_N=(\sum_{q=L-j+1}^L Y_{q}(N)-W^{(y,\pi,I,n,j)}_N)^{(+)} \ \mbox{and}
 \ \mbox{for} \ k<N,\\
&A_k=I(k,L-j+1,W^{(y,\pi,I,n,j)}_k-Y_{L-j+1}(k))\nonumber\\
&+E(J^{(\pi,I)}_{k+1}(U^{(\pi,I)}(W^{(y,\pi,I,n,j)}_k,k,j,\rho_1,...,
\rho_{k+1}),j-1,\rho_1,...,\rho_{k+1})|\mathcal{F}_k),\nonumber\\
&D_k=I(k,L-j+1,W^{(y,\pi,I,n,j)}_k-X_{L-j+1}(k))\nonumber\\
&+E(J^{(\pi,I)}_{k+1}(\tilde{U}^{(\pi,I)}(W^{(y,\pi,I,n,j)}_k,k,j,\rho_1,...,
\rho_{k+1}),j-1,\rho_1,...,\rho_{k+1})|\mathcal{F}_k).\nonumber
\end{eqnarray}
Set
\begin{equation}\label{5.13}
V_k=\min_{\sigma\in\Gamma_k}\max_{\tau\in\Gamma_k}E(D_{\sigma}
\mathbb{I}_{\sigma<\tau}+
A_{\tau}\mathbb{I}_{\tau\leq \sigma}|\mathcal{F}_k).
\end{equation}
Then for any $k\geq{n}$,
\begin{equation}\label{5.14}
V_k=J^{(\pi,I)}_k(W^{(y,\pi,I,n,j)}_k,k,j,\rho_1,...,\rho_k).
\end{equation}
Furthermore, the stopping times
\begin{equation}\label{5.15}
\begin{split}
\tilde\sigma=\tilde s_1=\tilde{s}(y,\pi,I,n,j)_1 \ \mbox{and} \
\tilde\tau=\tilde b_1=\tilde{b}(y,\pi,I,n,j)_1
\end{split}
\end{equation}
given by (\ref{5.10}) with $\tilde{s}(y,\pi,I,k,j)=(\tilde{s}_1,...
,\tilde{s}_j)$ and $\tilde{b}(y,\pi,I,k,j)=(\tilde{b}_1,...,\tilde{b}_j)$
satisfy
\begin{equation}\label{5.16}
E(D_{\tilde\sigma}\mathbb{I}_{\tilde\sigma<\tau}+A_{\tau}\mathbb{I}_{\tilde
\sigma\geq\tau}|\mathcal{F}_n)\leq
V_n\leq
E(D_{\sigma}\mathbb{I}_{\sigma<\tilde\tau}+A_{\tilde\tau}\mathbb{I}_{\sigma\geq
\tilde\tau}|\mathcal{F}_n)
\end{equation}
for any $\sigma,\tau\in\Gamma_n$.
\end{lem}
\begin{proof}
Fix $\pi,I\in\mathcal{A}\times\mathcal{I}$, and $j\geq{1}$. We will
use backward induction on $n$. For $n=N$ the statement is obvious
since all the terms in (\ref{5.14}) and (\ref{5.16}) are equal to
$((\sum_{i=L-j+1}^L f^{(i)}_N(\rho_1,...,\rho_N))-y)^{+}$. Suppose
that the assertion holds true for $n+1,...,N$ and prove it for $n$.
Fix $y\geq{0}$ and $n\leq k<N$ (for $k=N$ the statement is obvious).
Fix $m>k$ and denote $Z_m=W^{(y,\pi,I,n,j)}_{m}$. For any $i\geq{m}$
we have $W^{(Z_m,\pi,I,m,j)}_i=W^{(y,\pi,I,n,j)}_{i}$, and so
\begin{eqnarray}\label{5.17}
&A_N=(\sum_{q=L-j+1}^L Y_{q}(N)-W^{(Z_m,\pi,I,m,j)}_N)^{(+)} \ \mbox{and}
\ \mbox{for} \ m\leq{i}<N,\\
&A_i=I(i,L-j+1,W^{(Z_m,\pi,I,m,j)}_i-Y_{L-j+1}(i))\nonumber\\
&+E(J^{(\pi,I)}_{i+1}(U^{(\pi,I)}(W^{(Z_m,\pi,I,m,j)}_i,i,j,\rho_1,...,
\rho_{i+1}),j-1,\rho_1,...,\rho_{i+1})|\mathcal{F}_i),\nonumber\\
&D_i=I(i,L-j+1,W^{(Z_m,\pi,I,m,j)}_i-X_{L-j+1}(i))\nonumber\\
&+E(J^{(\pi,I)}_{i+1}(\tilde{U}^{(\pi,I)}(W^{(Z_m,\pi,I,m,j)}_i,i,j,\rho_1,
...,\rho_{i+1}),j-1,\rho_1,...,\rho_{i+1})|\mathcal{F}_i).\nonumber
\end{eqnarray}
Since $Z_m$ is $\mathcal{F}_{m}$-measurable then
 using the induction hypothesis for $m>k\geq{n}$ (with $Z_m$ in place of $y$)
  we obtain that for any $m>k$,
\begin{equation}\label{5.18}
V_m=J^{(\pi,I)}_m(Z_m,j,\rho_1,...,\rho_m).
\end{equation}
Thus
\begin{eqnarray}\label{5.18+}
&E(V_{k+1}|\mathcal{F}_k)=pJ^{(\pi,I)}_{k+1}(W^{(y,\pi,I,n,j)}_k+
\gamma(k,L-j+1,W^{(y,\pi,I,n,j)}_k)\\
&\times
S_0b\prod_{i=1}^{k}(1+\rho_i),j,\rho_1,...,\rho_k,b)\nonumber+(1-p)
J^{(\pi,I)}_{k+1}(W^{(y,\pi,I,n,j)}_k\nonumber\\
&+\gamma(k,L-j+1,W^{(y,\pi,I,n,j)}_k)S_0a\prod_{i=1}^{k}(1+\rho_i),j,
\rho_1,...,\rho_k,a)\nonumber.
\end{eqnarray}
Using Lemma \ref{lem3.1} for the processes ${\{A_i\}}_{i=k}^N$
and ${\{D_i\}}_{i=k}^N$ together with (\ref{5.3})-(\ref{5.6})
and (\ref{5.18+}) we obtain that for any $k\geq{n}$,
\begin{equation}\label{5.19}
V_k=J^{(\pi,I)}_k(W^{(y,\pi,I,n,j)}_k,1,\rho_k,...,\rho_k).
\end{equation}
From (\ref{5.10}) and (\ref{5.19}) it follows that
\begin{equation}\label{5.20}
\begin{split}
\tilde\sigma=N\wedge\min{\{i\geq{n}|V_i\geq D_i\}}, \
\tilde\tau=N\wedge\min{\{i\geq{n}|V_i=A_i\}}.
\end{split}
\end{equation}
Thus applying Lemma \ref{lem3.1} to the processes
${\{A_i\}}_{i=n}^N$ and ${\{D_i\}}_{i=n}^N$ we obtain (\ref{5.16}).
\end{proof}
\begin{lem}\label{lem4.1}
For any $\pi,I\in\mathcal{A}\times\mathcal{I}$, $n\leq{N}$,
$j\geq{1}$, $s,b\in\mathcal{S}^{(j)}_n$ and $y\geq{0}$,
\begin{eqnarray}\label{5.21}
&R(y,\pi,I,n,j,\tilde{s}(y,\pi,I,n,j),b)\leq R(y,\pi,I,n,j)\\
&=J^{(\pi,I)}_n(y,j,\rho_1,...,\rho_n)\leq
R(y,\pi,I,n,j,s,\tilde{b}(y,\pi,I,n,j))\nonumber.
\end{eqnarray}
\end{lem}
\begin{proof}
Fix $\pi,I\in\mathcal{A}\times\mathcal{I}$. We will use the backward
induction in $n$. For $n=N$ the statement is obvious since all the
terms are equal to $((\sum_{i=L-j+1}^L
f^{(i)}_N(\rho_1,...,\rho_N))-y)^{+}$. Suppose that the assertion is
correct for $n+1,...,N$ and let us prove it for $n$. For $j>1$,
$n\leq k_1<N$ and $k_2\in{\{0,1\}}$ define the map
$Q^{(k_1,k_2)}:\mathcal{S}^{(j)}_n\rightarrow\mathcal{S}^{(j-1)}_{k_1+1}$
by $Q^{(k_1,k_2)}(s_1,...,s_{i+1})=(s'_1,...,s'_i)$ where
\begin{eqnarray}\label{5.22}
&s'_1=s_2(k_1,k_2) \ \mbox{and} \
\mbox{for} \ m>1,\\
&s'_m((a_1,...,a_{m-1}),(d_1,...,d_{m-1}))\nonumber\\
&=s_{m+1}((k_1,a_1,...,a_{m-1}),(k_2,d_1,...,d_{m-1})).\nonumber
\end{eqnarray}
For any $j\geq{1}$ and $y\geq{0}$ set
$\tilde{s}=\tilde{s}(y,\pi,I,n,j)$. From (\ref{5.10})-(\ref{5.11})
it follows that for any $j>1$ the stopping strategy
$\tilde{s}^{(k_1,k_2)}=Q^{(k_1,k_2)}(\tilde{s})$ satisfies
\begin{equation}\label{5.23}
\tilde{s}^{(k_1,k_2)}=\tilde{s}(W^{(y,\pi,I,n,j,(k_1,k_2))}_{k_1+1},\pi,I,k_1+
1,j-1).
\end{equation}
Thus by the induction hypothesis we obtain that for any $n\leq
k_1<N$, $k_2\in {\{0,1\}}$, $j>1$ and
$b'\in\mathcal{S}^{(j)}_{k_1+1}$,
\begin{eqnarray}\label{5.24}
&R(W^{(y,\pi,I,n,j,(k_1,k_2))}_{k_1+1},\pi,I,k_1+1,j-1,\tilde{s}^{(k_1,k_2)},
b')\leq\\
&J^{(\pi,I)}_{k_1+1}(W^{(y,\pi,I,n,j,(k_1,k_2))}_{k_1+1},j-1,\rho_1,...,
\rho_{k_1+1}).\nonumber
\end{eqnarray}
Fix $j\geq{1}$, $y\geq{0}$ and let $b\in\mathcal{S}^{(j)}_n$. Set
$F(\tilde{s},b)=((\sigma_1,...,\sigma_j),(\tau_1,...,\tau_j))$,
$A={\{\sigma_1<\tau_1\}}$ and
$z=(\sigma_1\wedge\tau_1,\mathbb{I}_A)$. If $j>1$ denote also
$\tilde{s}'=\tilde{s}^{(\sigma_1\wedge\tau_1,\mathbb{I}_A)}$ and
$b'=\mathbb{I}_{\sigma_1\wedge\tau_1<N}Q^{(\sigma_1\wedge\tau_1,\mathbb{I}_A)}
(b)$ $+N\mathbb{I}_{\sigma_1\wedge\tau_1=N}$. In this case it follows
from (\ref{5.8}) that
\begin{eqnarray*}
&C(y,\pi,I,n,j,s,b)=\mathbb{I}_{j>1}\mathbb{I}_{\sigma_1\wedge\tau_1<N}
C(W^{(y,\pi,I,n,j,z)}_{\sigma_1\wedge\tau_1+1},\pi,I,\sigma_1\wedge\tau_1+1 ,
j-1, \tilde{s}',b')\\
&+I(\sigma_1\wedge\tau_1,L-j+1,W^{(y,\pi,I,n,j)}_{\sigma_1\wedge\tau_1}-
H^{(L-j+1)}(\sigma_1,\tau_1)).
\end{eqnarray*}
This together with (\ref{5.24}) gives
\begin{eqnarray}\label{5.25}
&R(y,\pi,I,n,j,\tilde{s},b)=\mathbb{I}_{j>1}E\bigg (E\big
(\mathbb{I}_{\sigma_1\wedge\tau_1<N}
C(W^{(y,\pi,I,n,j,z)}_{\sigma_1\wedge\tau_1
+1},\pi,I,\\
&\sigma_1\wedge\tau_1 +1,j-1, \tilde{s}',b')
|\mathcal{F}_{\sigma_1\wedge\tau_1+1}\big )|\mathcal{F}_n\bigg )+
E\big (I(\sigma_1\wedge\tau_1,L-j+1,\nonumber\\
&W^{(y,\pi,I,n,j)}_{\sigma_1\wedge\tau_1}-H^{(L-j+1)}(\sigma_1,\tau_1))|
\mathcal{F}_n\big )\nonumber\\
&=\mathbb{I}_{j>1}\times
E\big (\mathbb{I}_{\sigma_1\wedge\tau_1<N}R(W^{(y,\pi,I,n,j,z)}_{\sigma_1\wedge
\tau_1+1},\pi,I,
\sigma_1\wedge\tau_1 +1,j-1,\tilde{s}',b')|\mathcal{F}_n\big )\nonumber\\
&+E\big (I(\sigma_1\wedge\tau_1,L-j+1,W^{(y,\pi,I,n,
j)}_{\sigma_1\wedge\tau_1}-
H^{(L-j+1)}(\sigma_1,\tau_1))|\mathcal{F}_n\big )\nonumber\\
&\leq\mathbb{I}_{j>1}E\big (\mathbb{I}_{\sigma_1\wedge\tau_1<N}
J^{(\pi,I)}_{\sigma_1\wedge\tau_1 +1}
(W^{(y,\pi,I,n,j,z)}_{\sigma_1\wedge\tau_1
+1},j-1,\rho_1,...,\rho_{\sigma_1\wedge\tau_1 +1})
|\mathcal{F}_n\big )\nonumber\\
&+E\big (I(\sigma_1\wedge\tau_1,L-j+1,W^{(y,\pi,I,n,j)}_{\sigma_1\wedge\tau_1}-
H^{(L-j+1)}(\sigma_1,\tau_1))|\mathcal{F}_n\big ).\nonumber
\end{eqnarray}
Define the stochastic processes
${\{A_k\}}_{k=n}^N$ and ${\{D_k\}}_{k=n}^N$ by
\begin{eqnarray}\label{5.26}
&A_N=D_N=(\sum_{q=L-j+1}^L Y_{q}(N)-W^{(y,\pi,I,n,j)}_N)^{(+)} \ \mbox{and}
 \ \mbox{for} \ k<N,\\
&A_k=I(k,L-j+1,W^{(y,\pi,I,n,j)}_k-Y_{L-j+1}(k))\nonumber\\
&+E(J^{(\pi,I)}_{k+1}(U^{(\pi,I)}(W^{(y,\pi,I,n,j)}_k,k,j,\rho_1,...,
\rho_{k+1}),j-1,\rho_1,...,\rho_{k+1})|\mathcal{F}_k),\nonumber\\
&D_k=I(k,L-j+1,W^{(y,\pi,I,n,j)}_k-X_{L-j+1}(k))\nonumber\\
&+E(J^{(\pi,I)}_{k+1}(\tilde{U}^{(\pi,I)}(W^{(y,\pi,I,n,j)}_k,k,j,\rho_1,...,
\rho_{k+1}),j-1,\rho_1,...,\rho_{k+1})|\mathcal{F}_k).\nonumber
\end{eqnarray}
Observe that for any $\sigma,\tau\in\Gamma_n$,
\begin{eqnarray}\label{5.28}
&D_{\sigma}\mathbb{I}_{\sigma<\tau}+A_{\tau}\mathbb{I}_{\tau\leq\sigma}=
I(\sigma\wedge\tau,L-j+1,W^{(y,\pi,I,n,j)}_{\sigma\wedge\tau}\\
&-H^{(L-j+1)}(\sigma,\tau))+ E\big
(\mathbb{I}_{\sigma\wedge\tau<N}J^{(\pi,I)}_{\sigma\wedge\tau+1}
(W^{(y,\pi,I,n,j,z')}_{\sigma\wedge\tau
+1},j-1,\rho_1,...,\rho_{\sigma\wedge\tau+1})|\mathcal{F}_{\sigma\wedge\tau}
\big) \nonumber
\end{eqnarray}
where $z'=(\sigma\wedge\tau,\mathbb{I}_{\sigma<\tau})$. Thus
\begin{eqnarray}\label{5.28+}
&E(D_{\sigma}\mathbb{I}_{\sigma<\tau}+A_{\tau}\mathbb{I}_{\tau\leq\sigma}|
\mathcal{F}_n)=
E\big (I(\sigma\wedge\tau,L-j+1,W^{(y,\pi,I,n,j)}_{\sigma\wedge\tau}\\
&-H^{(L-j+1)}(\sigma,\tau))|\mathcal{F}_n\big )+
E\big (\mathbb{I}_{\sigma\wedge\tau<N}J^{(\pi,I)}_{\sigma\wedge\tau+1}
(W^{(y,\pi,I,n,j,z')}_{\sigma\wedge\tau+1}\nonumber\\
&,j-1,\rho_1,...,\rho_{\sigma\wedge\tau+1})|\mathcal{F}_n\big
).\nonumber
\end{eqnarray}
Since $\sigma_1=\tilde s_1=\tilde{s}(y,\pi,I,n,j)_1$ then from (\ref{5.25}),
(\ref{5.28+}) and Lemma \ref{lem4.0} it follows that for any
$b\in\mathcal{S}^{(j)}_n$,
\begin{eqnarray}\label{5.29}
&R(y,\pi,I,n,j,\tilde{s}(y,\pi,I,n,j),b)\leq
E(D_{\sigma_1}\mathbb{I}_{\sigma_1<\tau_1}+A_{\tau_1}\mathbb{I}_{\tau_1\leq
\sigma_1}|\mathcal{F}_n)\\
&\leq J^{(\pi,I)}_n(y,j,\rho_1,...,\rho_n).\nonumber
\end{eqnarray}
In a similar way we obtain that for any
$s\in\mathcal{S}^{(j)}_n$,
\begin{equation}\label{5.30}
R(y,\pi,I,n,j,s,\tilde{b}(y,\pi,I,n,j))\geq
J^{(\pi,I)}_n(y,j,\rho_1,...,\rho_n)
\end{equation}
completing the proof.
\end{proof}

In the final step we use Lemmas \ref{lem3.3} and \ref{lem3.4} and
Lemmas \ref{lem4.1} in order to construct an optimal hedge.
\begin{dfn}\label{dfn3.2}
Let $D\subset\mathbb{R}$ be an interval of the form $[a,b]$ or
$[a,\infty)$, $H$ be a set and $f:D\times{H}\rightarrow\mathbb{R}$
such that $f(\cdot,h)$ is a continuous function which has a minimum
on $D$. Define the function $argmin_{f}:H\rightarrow{D}$ by
$argmin_{f}(h)=\min\{y\in{D}|f(y,h)=\min_{ z\in{K}}f(z,h)\}$.
\end{dfn}
Lemmas \ref{lem3.3} and \ref{lem3.4} enable us to consider the
following functions. Define
$\tilde\gamma:{\{0,...,N-1\}}\times{\{1,...,L\}}\times\mathbb{R}\rightarrow{
\Xi}$
by
\begin{equation}\label{5.31}
\tilde\gamma(k,j,y)=\frac{argmin_f(k,j,\rho_1,...,\rho_k)}{S_0\prod_{i=1}^k
(1+\rho_i)}
\end{equation}
where
$f:K(y)\times{\{0,...,N-1\}}\times{\{1,...,L\}}\times{\{a,b\}^k}\rightarrow{
\mathbb{R}}$
is given by
\begin{eqnarray}\label{5.32}
&f(\alpha,k,j,u_1,...,u_k)=pJ_{k+1}(y+b\alpha,L-j+1,u_1,...,u_k,b) \\
&+(1-p)J_{k+1}(y+a\alpha,L-j+1,u_1,...,u_k,a)\nonumber.
\end{eqnarray}
Also define
$\tilde{I}:{\{0,...,N\}}\times{\{1,...,L\}}\times\mathbb{R}\rightarrow{\Xi}$
by
\begin{equation}\label{5.33}
\begin{split}
\tilde{I}(N,j,y)=((\sum_{i=j+1}^L Y_i(N)-y)^{+}, \
\tilde{I}(k,L,y)=(-y)^{+}.
\end{split}
\end{equation}
Then for $k<N$ and $j<L$,
\begin{equation}\label{5.34}
\tilde{I}(k,j,y)=argmin_g(k,j,\rho_1,...,\rho_k)
\end{equation}
where
$g:[-y^{+},\infty)\times{\{0,...,N-1\}}\times{\{1,...,L-1\}}\times{\{a,b\}^k}
\rightarrow{\mathbb{R}}$ is given by
\begin{eqnarray}\label{5.35}
&g(z,k,j,u_1,...,u_k)=z+\min_{\alpha\in K(y+z)}\big (pJ_{k+1}(y+z+b\alpha,L-j,
\\
&u_1,...,u_k,b)+(1-p)J_{k+1}(y+z+a\alpha,L-j,u_1,...,u_k,a)\big ).\nonumber
\end{eqnarray}
Clearly $\tilde{I}\in\mathcal{I}$. For any initial capital $x$
consider the portfolio strategy $\tilde{\pi}=(x,\tilde{\gamma})$.
Observe that $\tilde{\gamma}$ satisfies (\ref{2.5+}), and so
$\tilde{\pi}\in\mathcal{A}(x)$.
\begin{lem}\label{lem4.2}
For any $k\leq{N}$ and $(\pi,I)\in\mathcal{A}(x)\times\mathcal{I}$
\begin{equation}\label{5.36}
J_k(y,j,\rho_1,...,\rho_k)=J^{(\tilde{\pi},\tilde{I})}_k(y,j,\rho_1,...,\rho_k)
\leq
J^{(\pi,I)}_k(y,j,\rho_1,...,\rho_k).
\end{equation}
\end{lem}
\begin{proof}
We will use the backward induction. Fix
$\pi=(x,\gamma)\in\mathcal{A}(x)$ and $I\in\mathcal{I}$. For $k=N$
the statement is obvious. Suppose the assertion holds true for $n+1$
and prove it for $n$. For $j=0$ the statement is clear. Fix
$j\geq 1$. From the induction hypothesis and the definition of
$\tilde\gamma,\tilde{I}$ we obtain that
\begin{eqnarray}\label{5.37}
&\inf_{\alpha\in{K(y)}}\big (pJ_{n+1}(y+b\alpha,j,\rho_1,...,\rho_n,b) \\
&+(1-p)J_{n+1}(y+a\alpha,j,\rho_1,...,\rho_n,a)\big )\nonumber\\
&=pJ_{n+1}\big
(y+\tilde\gamma(n,L-j+1,y)S_0b\prod_{i=1}^{n}(1+\rho_i),j,u_1,
...,u_n,b\big )\nonumber \\
&+(1-p)J_{n+1}\big
(y+\tilde\gamma(n,L-j+1,y)S_0a\prod_{i=1}^{n}(1+\rho_i),j,u_1,
...,u_n,a\big )\nonumber\\
&=pJ^{(\tilde{\pi},\tilde{I})}_{n+1}\big (y+\tilde\gamma(n,L-j+1,y)
S_0b\prod_{i=1}^{n}(1+\rho_i),j,\rho_1,...,\rho_n,b\big )\nonumber \\
&+(1-p)J^{(\tilde{\pi},\tilde{I})}_{n+1}\big
(y+\tilde\gamma(n,L-j+1,y)
S_0\prod_{i=1}^{n}(1+\rho_i),j,\rho_1,...,\rho_n,a\big ).\nonumber
\end{eqnarray}
From the induction hypothesis and the fact that $\gamma$ satisfies
(\ref{2.5+}) it follows that
\begin{eqnarray}\label{5.38}
&\inf_{\alpha\in{K(y)}}\big (pJ_{n+1}(y+b\alpha,j,\rho_1,...,\rho_n,b)\\
&+(1-p)J_{n+1}(y+a\alpha,j,\rho_1,...,\rho_n,a)\big )\nonumber\\
&\leq\inf_{\alpha\in{K(y)}}\big (pJ^{(\pi,I)}_{n+1}(y+b\alpha,j,\rho_1,...,
\rho_n,b)\nonumber \\
&+(1-p)J^{(\pi,I)}_{n+1}(y+a\alpha,j,\rho_1,...,\rho_n,a)\big )\nonumber\\
&\leq
pJ^{(\pi,I)}_{n+1}(y+\gamma(n,L-j+1,y)S_0b\prod_{i=1}^{n}(1+\rho_i),j,
\rho_1,...,\rho_n,b)\nonumber \\
&+(1-p)J^{(\pi,I)}_{n+1}\big
(y+\gamma(n,L-j+1,y)S_0a\prod_{i=1}^{n}(1+\rho_i),j,
\rho_1,...,\rho_n,a\big ).\nonumber
\end{eqnarray}
From the induction hypothesis and the definition of
$\tilde{\gamma},\tilde{I}$ we obtain
\begin{eqnarray}\label{5.39}
&\inf_{z\geq(g^{(L-j+1)}_n(\rho_1,...,\rho_n)-y)^{+}}\inf_{\alpha\in{K(y+z-
g^{(L-j+1)}_n(\rho_1,...,\rho_n))}}\\
&\bigg (z+pJ_{n+1}\big (y+z-g^{(L-j+1)}_n(\rho_1,...,\rho_n)+
b\alpha,j-1,\rho_n,...,\rho_n,b\big )\nonumber\\
&+(1-p)J_{n+1}\big (y+z-g^{(L-j+1)}_n(\rho_1,...,\rho_n)+a\alpha,j-1,\rho_1,
...,\rho_n,a\big )\bigg )\nonumber\\
&=\tilde{I}\big (n,L-j+1,y-g^{(L-j+1)}_n(\rho_1,...,\rho_n)\big)\nonumber\\
&+pJ_{n+1}\big (\tilde{U}^{(\tilde\pi,\tilde{I})}(y,n,j,\rho_1,...,\rho_n,b),
j-1,\rho_1,...,\rho_n,b\big )\nonumber \\
&+(1-p)J_{n+1}\big (\tilde{U}^{(\tilde\pi,\tilde{I})}(y,n,j,\rho_1,...,\rho_n,
a),j-1,\rho_1,...,\rho_n,a\big )\nonumber\\
&=\tilde{I}\big (n,L-j+1,y-g^{(L-j+1)}_n(\rho_1,...,\rho_n)\big )+\nonumber\\
&pJ^{(\tilde\pi,\tilde{I})}_{n+1}\big (\tilde{U}^{(\tilde\pi,\tilde{I})}(y,n,j,
\rho_1,...,\rho_n,b),j-1,\rho_1,...,\rho_n,b\big )\nonumber \\
&+(1-p)J^{(\tilde{\pi},\tilde{I})}_{n+1}\big (\tilde{U}^{(\tilde\pi,\tilde{I})}
(y,n,j,\rho_1,...,\rho_n,a),j-1,\rho_1,...,\rho_n,a\big ).
\nonumber
\end{eqnarray}
Using that $\gamma$ satisfies (\ref{2.5+}) and
$I(\cdot,\cdot,u)\geq{(-u)}^{+}$ it follows by the induction hypothesis that
\begin{eqnarray}\label{5.40}
&\inf_{z\geq(g^{(L-j+1)}_n(\rho_1,...,\rho_n)-y)^{+}}\inf_{\alpha\in{K(y+z-
g^{(L-j+1)}_n(\rho_1,...,\rho_n))}}\\
&\bigg (z+pJ_{n+1}\big (y+z-g^{(L-j+1)}_n(\rho_1,...,\rho_n)+b\alpha,j-1,
\rho_n,...,\rho_n,b)\nonumber\\
&+(1-p)J_{n+1}\big (y+z-g^{(L-j+1)}_n(\rho_1,...,\rho_n)+a\alpha,j-1,\rho_1,
...,\rho_n,a\big )\bigg )\nonumber\\
&\leq I\big (n,L-j+1,y-g^{(L-j+1)}_n(\rho_1,...,\rho_n)\big )\nonumber\\
&+pJ_{n+1}\big (\tilde{U}^{(\pi,I)}(y,n,j,\rho_1,...,\rho_n,b),j-1,\rho_1,
...,\rho_n,b\big )\nonumber \\
&+(1-p)J_{n+1}\big (\tilde{U}^{(\pi,I)}(y,n,j,\rho_1,...,\rho_n,a),j-1,\rho_1,
...,\rho_n,a\big )\nonumber\\
&\leq I\big (n,L-j+1,y-g^{(L-j+1)}_n(\rho_1,...,\rho_n))\nonumber\\
&+pJ^{(\pi,I)}_{n+1}(\tilde{U}^{(\pi,I)}(y,n,j,\rho_1,...,\rho_n,b),j-1,
\rho_1,...,\rho_n,b\big )\nonumber \\
&+(1-p)J^{(\pi,I)}_{n+1}\big (\tilde{U}^{(\pi,I)}(y,n,j,\rho_1,...,\rho_n,a),
j-1,\rho_1,...,\rho_n,a\big ).
\nonumber
\end{eqnarray}
In a similar way we obtain
\begin{eqnarray}\label{5.41}
&\inf_{z\geq(f^{(L-j+1)}_n(\rho_1,...,\rho_n)-y)^{+}}\inf_{\alpha\in{K(y+z-
f^{(L-j+1)}_n(\rho_1,...,\rho_n))}}\\
&\bigg (z+pJ_{n+1}\big (y+z-f^{(L-j+1)}_n(\rho_1,...,\rho_n)+b\alpha,j-1,
\rho_n,...,\rho_n,b\big )\nonumber\\
&+(1-p)J_{n+1}\big (y+z-f^{(L-j+1)}_n(\rho_1,...,\rho_n)+a\alpha,j-1,\rho_1,
...,\rho_n,a\big )\bigg )\nonumber\\
&=\tilde{I}\big (n,L-j+1,y-f^{(L-j+1)}_n(\rho_1,...,\rho_n)\big )\nonumber\\
&+pJ_{n+1}\big (U^{(\tilde\pi,\tilde{I})}(y,n,j,\rho_1,...,\rho_n,b),j-1,
\rho_1,...,\rho_n,b\big )\nonumber \\
&+(1-p)J_{n+1}\big (U^{(\tilde\pi,\tilde{I})}(y,n,j,\rho_1,...,\rho_n,a),j-1,
\rho_1,...,\rho_n,a\big )\nonumber\\
&=\tilde{I}\big (n,L-j+1,y-f^{(L-j+1)}_n(\rho_1,...,\rho_n)\big )\nonumber\\
&+pJ^{(\tilde\pi,\tilde{I})}_{n+1}\big (U^{(\tilde\pi,\tilde{I})}(y,n,j,\rho_1,
...,\rho_n,b),j-1,\rho_1,...,\rho_n,b\big )\nonumber \\
&+(1-p)J^{(\tilde{\pi},\tilde{I})}_{n+1}\big (U^{(\tilde\pi,\tilde{I})}(y,n,j,
\rho_1,...,\rho_n,a),j-1,\rho_1,...,\rho_n,a\big )
\nonumber
\end{eqnarray}
and
\begin{eqnarray}\label{5.42}
&\inf_{z\geq(f^{(L-j+1)}_n(\rho_1,...,\rho_n)-y)^{+}}\inf_{\alpha\in{K(y+z-
f^{(L-j+1)}_n(\rho_1,...,\rho_n))}}\\
&\bigg (z+pJ_{n+1}\big (y+z-f^{(L-j+1)}_n(\rho_1,...,\rho_n)+b\alpha,j-1,
\rho_n,...,\rho_n,b\big )\nonumber\\
&+(1-p)J_{n+1}\big (y+z-f^{(L-j+1)}_n(\rho_1,...,\rho_n)+a\alpha,j-1,\rho_1,
...,\rho_n,a\big )\bigg )\nonumber\\
&\leq I\big (n,L-j+1,y-f^{(L-j+1)}_n(\rho_1,...,\rho_n)\big )\nonumber\\
&+pJ_{n+1}\big (U^{(\pi,I)}(y,n,j,\rho_1,...,\rho_n,b),j-1,\rho_1,...,\rho_n,
b\big )\nonumber \\
&+(1-p)J_{n+1}\big (U^{(\pi,I)}(y,n,j,\rho_1,...,\rho_n,a),j-1,\rho_1,...,
\rho_n,a\big )\nonumber\\
&\leq I\big (n,L-j+1,y-f^{(L-j+1)}_n(\rho_1,...,\rho_n)\big )\nonumber\\
&+pJ^{(\pi,I)}_{n+1}\big (U^{(\pi,I)}(y,n,j,\rho_1,...,\rho_n,b),j-1,\rho_1,
...,\rho_n,b\big )\nonumber \\
&+(1-p)J^{(\pi,I)}_{n+1}\big (U^{(\pi,I)}(y,n,j,\rho_1,...,\rho_n,a),j-1,
\rho_1,...,\rho_n,a\big ).
\nonumber
\end{eqnarray}
Now, (\ref{5.36}) follows from (\ref{5.37})--(\ref{5.42}).
\end{proof}
Finally, fix an initial capital $x\geq{0}$ and let
$\tilde{\pi}=(x,\tilde\gamma)$. Set
\begin{equation}\label{5.43}
\tilde{s}=\tilde{s}(x,\tilde\pi,\tilde{I},0,L).
\end{equation}
Using Lemmas \ref{lem4.1} and \ref{lem4.2} (for $j=L$ and $n=0$) we
obtain that for any
$\pi,I,s\in\mathcal{A}(x)\times\mathcal{I}\times{S}$,
\begin{equation*}
R(\tilde\pi,\tilde{I},\tilde{s})=J^{(\tilde\pi,\tilde
I)}_0(x,L)=J_0(x,L)\leq J^{(\pi, I)}_0(x,L)=R(\pi,I,s).
\end{equation*}
Thus
\begin{equation*}
R(\tilde\pi,\tilde{I},\tilde{s})=R(x)=J^{(\pi,I)}_0(x,L).
\end{equation*}
completing the proof of Theorem \ref{thm2.2}. \qed

\end{document}